\documentclass[prx,aps,floatfix,amsmath,amssymb,superscriptaddress,tightenlines,twocolumn]{revtex4-1}
\usepackage{hyperref}
\hypersetup{
  citecolor=red,
    colorlinks=true,
    linkcolor=blue,
    filecolor=magenta,
    urlcolor=cyan,
}
\usepackage{algorithm}
\usepackage{algorithmic}
\usepackage{graphicx}
\usepackage{epstopdf}
\newcommand{\be}{\begin{equation}}
\newcommand{\ee}{\end{equation}}
\usepackage{amsmath}
\usepackage{amsfonts}
\usepackage{bm}
\usepackage{color}
\usepackage{subfigure}
\usepackage{amsthm}
\usepackage[latin1]{inputenc}
\usepackage{tikz}
\usepackage{comment}

\newcommand{\bra}[1]{\left\langle #1 \right|}
\newcommand{\ket}[1]{\left|#1\right\rangle}

\newcommand{\Tr}{\textnormal{Tr}}

\newcommand{\mc}{\Gamma}

\newtheorem{lem}{Lemma}
\newtheorem{thm}{Theorem}

\begin{document}
\title{Robust entanglement renormalization on a noisy quantum computer}
\author{Isaac H. Kim}
\affiliation{IBM T.J. Watson Research Center, Yorktown Heights, NY, 10598, USA}
\affiliation{Stanford Institute for Theoretical Physics, Stanford University, Stanford CA 94305 USA}
\author{Brian Swingle}
\affiliation{Condensed Matter Theory Center, Maryland Center for Fundamental Physics, \\
Joint Center for Quantum Information and Computer Science, \\
and Department of Physics, University of Maryland, College Park, MD 20742, USA}
\affiliation{Kavli Institute for Theoretical Physics, Santa Barbara, CA 93106, USA}
\date{\today}
\begin{abstract}
  A method to study strongly interacting quantum many-body systems at and away from criticality is proposed. The method is based on a MERA-like tensor network that can be efficiently and reliably contracted on a noisy quantum computer using a number of qubits that is much smaller than the system size. We prove that the outcome of the contraction is stable to noise and that the estimated energy upper bounds the ground state energy. The stability, which we numerically substantiate, follows from the positivity of operator scaling dimensions under renormalization group flow. The variational upper bound follows from a particular assignment of physical qubits to different locations of the tensor network plus the assumption that the noise model is local. We postulate a scaling law for how well the tensor network can approximate ground states of lattice regulated conformal field theories in $d$ spatial dimensions and provide evidence for the postulate. Under this postulate, a $O(\log^{d}(1/\delta))$-qubit quantum computer can prepare a valid quantum-mechanical state with energy density $\delta$ above the ground state. In the presence of noise, $\delta = O(\epsilon \log^{d+1}(1/\epsilon))$ can be achieved, where $\epsilon$ is the noise strength.

\end{abstract}
\maketitle
\section{Introduction}
Recently, there has been an impressive amount of growth in quantum technology. Planar superconducting qubit architectures with error rates below the fault tolerance threshold~\cite{Raussendorf2007} have been reported~\cite{Barends2015,Sheldon2016}. Ion traps have demonstrated an error rate that is even an order of magnitude lower~\cite{Ballance2016}. Qubits based on topologically protected Majorana fermions have been reported as well \cite{Mourik2012}. If these devices can be scaled up while maintaining error rates below the fault tolerance threshold, it would be possible to construct a large-scale fault tolerant quantum computer.

These are encouraging developments, but we should be mindful of the remaining challenges. In order to perform fault tolerant quantum computation, one necessarily needs to incur a rather large error correction overhead. In the the leading surface code architecture~\cite{Raussendorf2007}, the overhead scales polylogarithmically with the size of the computation. This amounts to a modest increase in the number of requisite physical qubits, in the asymptotic limit in which the size of the computation becomes large. However, for solving practical problems of interest, the estimated number of extra qubits usually is a few orders of magnitude larger than the number of requisite logical qubits. For example, in order to break the existing RSA-2048 cryptosystem, assuming a physical noise rate of $10^{-3}$, one would need roughly $10^3$ physical qubits per logical qubit \cite{Fowler2012}. This is likely to pose a practical challenge in implementing large-scale quantum algorithms in the near term.

Until we overcome these challenges, we will be left with devices that are too large to classically simulate, yet not large enough to implement full-scale fault tolerant quantum computation. Can we use nevertheless these devices to solve any outstanding problems in physics?

We believe there are numerous opportunities in this direction, especially for studying strongly interacting quantum many-body systems at low energy. Specifically, we would like to argue that such a noisy quantum device can be used as a highly efficient machine for computing the energy in variational calculations; see FIG.~\ref{fig:paradigm}. In this paradigm, we view the quantum device as an abstract machine from which expectation values of various observables, e.g., energy or magnetization, can be measured. The measured energy is fed into a classical optimizer. The optimizer updates the parameters of the quantum device to lower the energy. This process is repeated until convergence.
\begin{figure}[h]
  \begin{tikzpicture}
    \node[align = center, rectangle, draw, rounded corners, fill=red!20, text width=2cm, minimum height = 1cm] (quantum) at (0,0){Quantum Processor };
    \node[align = center, rectangle, draw, rounded corners, fill=green!20, text width=2cm, minimum height = 1cm] (classical) at (4,0){Classical \\ Optimizer};
    \draw[->] (classical) to[out=150, in=30] node[midway,above] {\begin{tabular}{c} Energy \\ Lowered \end{tabular}} (quantum);
    \draw[->] (quantum) to[out=-30, in=-150] node[midway,below] {\begin{tabular}{c} Energy \\ Measured \end{tabular}} (classical);
  \end{tikzpicture}
  \caption{Energy estimated from a quantum processor is fed into a classical computer. Based on the measured values of energy at previous iterations, the classical computer updates the parameter of the quantum processor. \label{fig:paradigm}}
\end{figure}
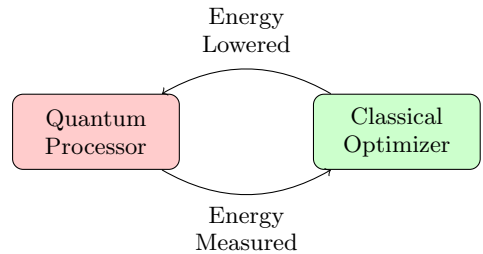

This paradigm originated from the quantum chemistry community \cite{Peruzzo2013}; see also Ref.~\cite{Wecker2015} for a related work on the Hubbard model. In their context, a quantum processor consisting of $n$ qubits represents a state of a molecule consisting of $O(n)$ orbitals. A variational state is prepared by applying a low-depth quantum circuit to a product state. Then a classical optimization method is employed to iteratively lower the energy.

Provided that the measurements are reliable, the obtained energy must be variational, even in the presence of decoherence and systematic errors. The rationale is simple. The circuit, whether noisy or not, implements a physically allowed operation. As such, the state of the device is a valid physical state and the variational principle applies. Accounting for the measurement error, we can conclude that
\begin{equation}
  E_0 \leq E_M + \epsilon_M, \label{eq:variational}
\end{equation}
where $E_0$ is the ground state energy, $E_M$ is the estimated energy, and $\epsilon_M$ quantifies the measurement error. One should view $\epsilon_M$ as an analogue of the numerical precision in classical variational methods. For modern classical computers, this number is often too small to be of any significance. For quantum computers, the present-day magnitude of $\epsilon_M$ can be small, but it is certainly not negligible.

We believe that, for a class of quantum many-body systems that have so far defied efficient classical simulation, a near-term quantum device consisting of around $50$ qubits with error rate of $\epsilon = 10^{-3}$ can report a variational upper bound that is lower than any upper bound obtained from existing classical variational methods. This means that such a device must be capable of representing a valid quantum-mechanical state of a much larger system. The device state must be able to approximate a wide variety of physical states that are difficult to simulate with present-day computing power, and furthermore, preparation of such states must be somehow resilient to noise. Moreover, there must be a decisive advantage in using a noisy quantum computer over a classical computer. These requirements are a tall order; can they possibly be satisfied?

Surprisingly, the answer is yes. In this paper, we propose a method that, in our opinion, can overcome these challenges. The main idea is to use a tensor network that is tailor-made to fulfill all of the above requirements. We call this tensor network a deep MERA (DMERA); it is a variant of Vidal's multi-scale entanglement renormalization ansatz (MERA) \cite{Vidal2008}. We show that DMERA possesses a number of attractive qualities, which we summarize as follows.

First, DMERA can be contracted extremely fast on a quantum computer. Specifically, for a variational parameter $D$, a quantum computer can compute local observables in $O(D\log N)$ time using $O(D^d)$ qubits, where $N$ is the system size and $d$ is the number of spatial dimensions. In contrast, a classical computer requires memory and computation time that scale exponentially with $D^d$. The dependence on $N$ is the same, but we emphasize that the exponential dependence on $D^d$ makes a classical simulation rather difficult. For example, we estimate that a classical simulation of a $d=2$ dimensional system with a modest $D=2$ would require simulating high depth circuits acting on mixed states of $64$ qubits. Given that a low-depth simulation of a $56$ qubit system requires at least a few terabytes of memory~\cite{Pednault2017}, the requisite resources for classical simulation seem to be far greater than what is available with even very powerful supercomputers. Assuming that the gate time is in the order of microseconds, the same computation can be carried out on a quantum computer in a few seconds or less.

Second, DMERA can approximate a wide variety of physical ground states using a small number of parameters, meaning small $D$. While we do not have a general theorem that guarantees this postulate, we provide a long list of evidence that suggests that the approximation error for local expectation values decays exponentially with $D$. In particular, we argue that a wide variety of topological states, quantum critical points, and even lattice regulated holographic quantum field theories have this property. Assuming this scaling, in order to reach a precision of $\delta$ for local observables, it suffices to choose $D$ to be $O(\log 1/\delta)$.

Third, contraction of DMERA on a quantum computer is resilient to noise. Suppose the quantum computer suffers from a noise rate of at most $\epsilon$ for every gate, preparation, and measurement. For the physical states we study, the expectation values of local observables are shown to be altered at most by $O(\epsilon D^{d+1})$, independent of $N$. Our assumption is that the lowest scaling dimension of the underlying system is positive. This is generically true for DMERA, and is likely to be true for unitary scale invariant field theories. We prove the stability rigorously, and numerically substantiate it to understand the typical influence of noise.

Fourth, given a DMERA ansatz and a local Hamiltonian, its energy can be estimated in such a way that Eq.~\ref{eq:variational} holds even in the presence of noise. Our only assumption is that every ideal gate can be approximated by its noisy counterpart acting on the same qubits. By measuring the energy this way, one can objectively compare the measured energy to the energy obtained from classical variational calculations.

These are similar to the qualities that have been advocated in a previous proposal by one of us \cite{Kim2017,Kim2017a}. While the underlying details vary, they share the same philosophy: to devise a variational method that is practical and resilient to noise. The fact that a minor modification of an existing tensor network gives rise to the aforementioned desirable qualities is encouraging.

Our work points to the possibility of fostering a symbiotic relationship between tensor networks and near-term quantum computers. Our understanding of tensor network simulation \cite{Verstraete2004,Vidal2008} has been developing rapidly \cite{Swingle2014,Evenbly2015,Czech2016,Swingle2016,Haegeman2017}, but our limited ability to manipulate large tensors on a classical computing device has severely hindered classical simulation of quantum many-body physics, especially in more than one dimension. On the other hand, near-term quantum computers are likely to suffer from noise and be of small size. This poses a challenge in implementing well-known quantum algorithms, e.g., factoring \cite{Shor1997}. It seems that these two different technologies can be merged together in a way that compensates for their individual weaknesses; the result is a kind of quantum assisted tensor network technology.

The rest of the paper is structured as follows. In Section~\ref{section:proposal}, we define DMERA and explain how one can variationally find an ansatz that approximates the ground state of a local Hamiltonian. In Section~\ref{section:faithfulness}, we argue that a large class of physical states can be well-approximated by a DMERA with a small number of parameters. In Section~\ref{section:noise_resilience}, we show that the outcome of the contraction is resilient to noise. In Section~\ref{section:contraction}, we explain how the network can be contracted on a quantum computer. In particular, we show that the energy estimated from this contraction sequence is variational. We discuss some potential applications in Section~\ref{section:applications}.

\section{Proposal \label{section:proposal}}

We propose a variational ansatz which we refer to as \emph{deep multi-scale entanglement enormalization ansatz}, or DMERA for short. DMERA is a version of the well-known multi-scale entanglement renormalization ansatz (MERA) \cite{Vidal2008}, and like MERA, it is a special kind of tensor network composed of unitary and isometric tensors such that the network can also be viewed as a quantum circuit. As in MERA, the key idea of DMERA is to disentangle local degrees of freedom which can then be removed using isometries, i.e., unitaries with one input fixed to a product state. The main difference is in the way the disentangling operation is carried out.

In MERA, say in the context of a one-dimensional lattice of qubits, one groups the individual qubits into clusters with effective dimension called the bond dimension. Each scale of the MERA then consists of one layer of unitaries and one layer of isometries. The variational parameters are contained in the unitaries and isometries and the number of variational parameters is determined by the bond dimension. In DMERA, rather than grouping qubits into clusters of some desired bond dimension, we instead allow each scale to consist of many layers of two-qubit unitaries. The variational parameters are still contained in the unitaries, but the number of variational parameters is now determined by the depth $D$ of the circuit at each scale. Any DMERA can be realized as a MERA with sufficiently large bond dimension; similarly, DMERA can approximate any MERA given sufficiently large depth $D$.

\subsection{DMERA}
Let us formally define DMERA for one-dimensional systems. A state $\ket{\psi}$ over $L=2^n$ qubits is a DMERA with depth $D$ if there exists a sequence of states $\{|\psi_i\rangle\}$ such that $\ket{\psi} = \ket{\psi_n}$, $\ket{\psi_0} =|0\rangle$, and
\begin{equation} \label{eq:dmerastep}
  \ket{\psi_{i+1}} = U_i \left[ \ket{\psi_{i}} \otimes \underbrace{\ket{0\ldots 0}}_{2^{i}}\right],
\end{equation}
where $U_i$ is a depth $D$ local quantum circuit consisting of two-qubit gates and the $2^i$ $\ket{0}$ are interspersed between the qubits that have been introduced at $j\leq i$; see FIG.~\ref{fig:dmera}. The gates can be labeled in terms of the pair of qubits that they act on and the time at which they are implemented. There are $n =\log_2 L$ renormalization steps and in each step we have $D$ layers of unit depth unitaries. Therefore, every gate can be specified in terms of a pair of tuples $(i,j)$ and $(s,y)$ where $i$ and $j$ are the qubits that the gate acts on, $s \in \{1,\ldots, n \}$ is the renormalization step, and $y \in \{1,\ldots, D \}$ specifies the layer within the renormalization step.
\begin{figure}[h]
  \begin{tikzpicture}
    \draw[] (0,0) -- (0,4) -- (1,2) -- cycle;
    \draw[] (0.5,2) node[]{$\psi_{i+1}$};
    \foreach \y in {1,...,8}
    {
      \draw[] (0, \y * 0.5-0.25) -- (-0.5, \y * 0.5-0.25);
    }
    \draw[] (1.5,2) node[]{$=$};
    \begin{scope}[xshift=5.5cm]
      \draw[] (0,0) -- (0,4) -- (1,2) --cycle;
      \draw[] (0.5,2) node[]{$\psi_i$};
      \foreach \y in {1,...,4}
      {
        \draw[] (0, \y -0.25) -- (-.75, \y -0.25);
        \draw[] (-0.25, \y-0.75) node[]{$\ket{0}$};
        \draw[] (-0.5, \y - 0.75) -- (-.75, \y - 0.75);
      }
      \foreach \x in {0,2,6}
      {
        \foreach \y in {1,...,4}
        {
          \filldraw[black] (-1-0.375*\x, \y - 0.875) rectangle ++ (0.25,0.75);
          \draw[] (-1-0.375*\x, \y - 0.25) -- (-1.25-0.375*\x, \y - 0.25);
          \draw[] (-1-0.375*\x, \y - 0.75) -- (-1.25-0.375*\x, \y - 0.75);
          \draw[] (-1-0.375*\x+0.5, \y - 0.25) -- (-1.25-0.375*\x+0.5, \y - 0.25);
          \draw[] (-1-0.375*\x+0.5, \y - 0.75) -- (-1.25-0.375*\x+0.5, \y - 0.75);
        }
      }
      \foreach \x in {1}
      {
        \foreach \y in {1,...,3}
        {
          \filldraw[black] (-1-0.375*\x, \y - 0.375) rectangle ++ (0.25,0.75);
          \draw[] (-1-0.375*\x, \y - 0.25) -- (-1.25-0.375*\x, \y - 0.25);
          \draw[] (-1-0.375*\x, \y - 0.75) -- (-1.25-0.375*\x, \y - 0.75);
          \draw[] (-1-0.375*\x+0.5, \y - 0.25) -- (-1.25-0.375*\x+0.5, \y - 0.25);
          \draw[] (-1-0.375*\x+0.5, \y - 0.75) -- (-1.25-0.375*\x+0.5, \y - 0.75);
        }
      }
      \draw[] (-2.375,2) node[]{$\cdots$};
    \end{scope}
  \end{tikzpicture}
  \caption{Recursion relation defining a DMERA. The state $\ket{\psi_{i+1}}$ is constructed from $\ket{\psi_i}$ by placing ancillas and applying a depth-$D$ local quantum circuit consisting of two-qubit gates.  \label{fig:dmera}}
\end{figure}
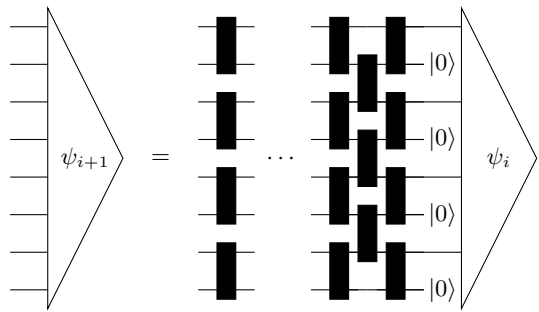

An important property of DMERA is that expectation values of local observables can be computed in time $O(e^{O(D)}\log L \log (1/\eta))$ on a classical computer and time $O(\frac{D\log L}{\eta^2})$ on a quantum computer, where $\eta$ is the desired precision. To see why, it is convenient to recall the notion of a \emph{past causal cone} \cite{Giovannetti2008,Evenbly2009}. Given an observable $\hat{O}$, its past causal cone is the set of gates that can influence its expectation value. In particular, the width of the past causal cone determines the number of physical qubits that are sufficient to estimate the expectation value. This is because of the following recursion relation:
\begin{equation}
  \bra{\psi_{i+1}} \hat{O} \ket{\psi_{i+1}} = \bra{\psi_{i}} \Phi_i(\hat{O}) \ket{\psi_i},
\end{equation}
where
\begin{equation}
  \Phi_i(\cdot) = \underbrace{\bra{0,\ldots,0}}_{2^i} U_i^{\dagger} \cdot U_i \underbrace{\ket{0,\ldots,0}}_{2^i}
\end{equation}
is a unital quantum channel that preserves locality \cite{Kim2017b}.

Specifically, let $s_i$ be the size of the nontrivial support of the operator $\Phi_{n-i} \circ \cdots \circ\Phi_{n-1}(\hat{O})$, and suppose $\hat{O}$ is supported on a finite interval of length $s_0=\ell$. Then we have the following bound:
\begin{equation}
  s_{i+1} \leq \frac{s_i + 2D}{2}.
\end{equation}
Iterating this for $i > 0$ leads to
\begin{equation}
s_i \leq \frac{\ell}{2^i} + D \sum_{k=0}^{i-1} \frac{1}{2^k} \leq \frac{\ell}{2^i} + 2D,
\end{equation}
which subsequently implies that the operator $U_{n-i-1}^\dagger \Phi_{n-i}\cdots \circ\Phi_{n-1}(\hat{O}) U_{n-i-1}$, i.e., the operator before we project onto $|0...0\rangle$ at scale $n-i-1$, acts on at most $\frac{\ell}{2^i} + 4D$ qubits. The expectation value can be computed on a classical computer by multiplying matrices whose dimension grows exponentially with $D$. This exponential cost can be removed on a quantum computer by simply implementing the gates. A similar conclusion holds in higher dimensions as well. For a depth $D$ DMERA in $d$ spatial dimensions, a classical computer can compute expectation values by multiplying matrices whose dimension grows exponentially with $D^d$, while on a quantum computer one needs only $O(D^d)$ physical qubits to compute the expectation value.

\subsection{Energy minimization protocol}
A method to contract the tensor network on a quantum computer will be discussed in Section~\ref{section:contraction}. For now, we assume that the expectation values of local observables can be estimated. Our goal is to understand how, given such a subroutine, one can find an approximate ground state. Our protocol is analogous to the one used in classical MERA calculations. The idea is to minimize the energy with respect to each of the tensors sequentially until the energy cannot be lowered anymore. This approach is heuristic and may get stuck in local minima. However, in practice this method has been shown to work well \cite{Evenbly2009}.

Motivated by this observation, we propose to minimize the energy of a DMERA as follows. Without loss of generality, let us consider a single circuit element $U$ in DMERA. One can easily show that the energy depends quadratically on $U$:
\begin{equation}
  E(U) = \sum_i \Tr(\rho_i U h_i U^{\dagger}),\label{eq:quadratic_form}
\end{equation}
where the $\rho_i$ are positive semi-definite operators and the $h_i$ are Hermitian matrices. Then the problem boils down to whether, given the measurement outcome $E(U)$, one can minimize it. This is a finite-dimensional classical optimization problem with a noisy input. The input is noisy because the value of $E(U)$ is sampled from some distribution, as opposed to being computed explicitly; see Section~\ref{section:contraction} for the details.

Therefore, we should employ a classical optimization method which, given a noisy objective function, minimizes it. The so-called simultaneous perturbation stochastic approximation (SPSA) \cite{Spall1992} is one such method. Given an objective function, SPSA estimates the gradient in a random direction and updates the variables in the opposite direction. Let $\Lambda = (\lambda_1, \ldots, \lambda_n)$ be a vector consisting of the variational parameters. In our setup, they parametrize the unitaries. The algorithm proceeds as follows:
\begin{enumerate}
  \item Initialize $\Lambda$ with random entries.
  \item Pick a random vector $v=(v_1, \ldots, v_n)$, where $v_i$ are iid random variables with $E[1/v_i]<\infty$.\footnote{An example would be a $v_i=\pm 1$ with probability $\frac{1}{2}$ for both outcomes. Note that the normal distribution does not satisfy this condition.}
  \item Estimate the gradient by computing $g_k = (E(\Lambda + \alpha_k) - E(\Lambda - \alpha_k v )) / (2\alpha_k)$.
  \item Set $\Lambda \leftarrow \Lambda - g_k \beta_k v $.
  \item Repat 2-4 until convergence.
\end{enumerate}

Two parameters, $\alpha_k$ and $\beta_k$, whose optimal values depend on the underlying problem, determine the rate of convergence. At the $k$-th iteration, these parameters are chosen to be
\begin{equation}
  \begin{aligned}
    \alpha_k &= \frac{a}{(k+1+A)^s}\\
    \beta_k &= \frac{b}{(k+1)^t},
  \end{aligned}
\end{equation}
where $a,b,A,s,t$ are constants. Asymptotically optimal values are $s=1$ and $t=1/6$, but these choices may not be optimal in practice.

What makes SPSA attractive is that the method continues to perform well even with noisy measurements of the objective function. This fact has been observed and utilized in the context of quantum state tomography~\cite{Granade2016}. In our own numerical experiment, we have observed a similar tendency. In order to assess the performance of this algorithm, we picked random choices of $\rho_i$ and $h_i$. While these operators will not be completely random in practice, the fact that the objective function depends quadratically on $U$ remains the same. This is our justification for using this numerical experiment as a proxy for the performance we expect in real experiments. We chose the parameters as $a=0.05$, $b=0.01,$ and $A=10$. We add an artificial stochastic noise to the measurement outcome which is $1\%$ of the operator norm of $h_i$. The results are plotted in FIG.~\ref{fig:spsa_superimposed} (100 samples), FIG.~\ref{fig:spsa_averaged} (averaged over samples), and FIG.~\ref{fig:spsa_worst} (worst case).

\begin{figure}[h]
  \includegraphics[width=\columnwidth]{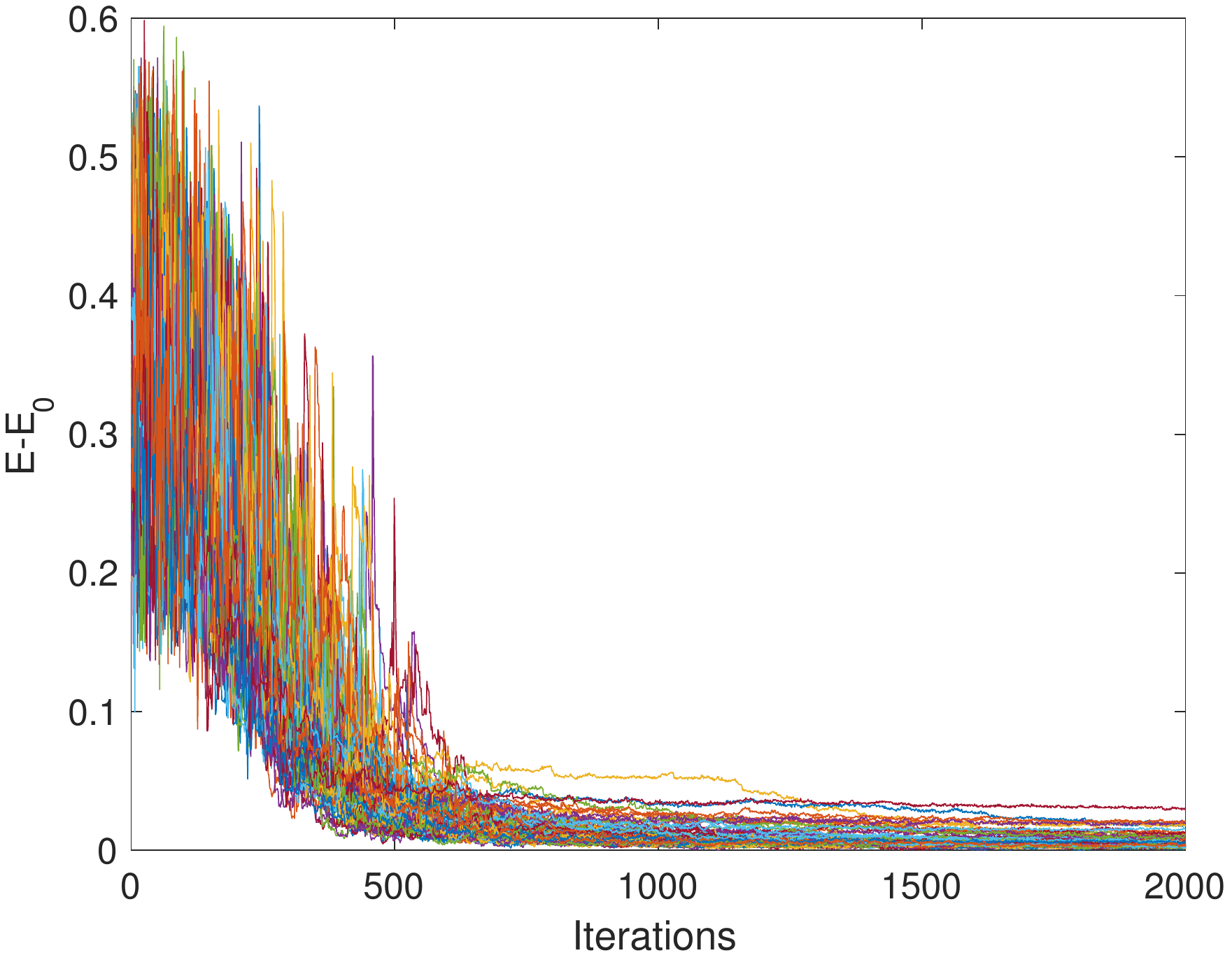}
  \caption{Energy minimization, 100 samples superimposed. \label{fig:spsa_superimposed}}
\end{figure}
\begin{figure}[h]
  \includegraphics[width=\columnwidth]{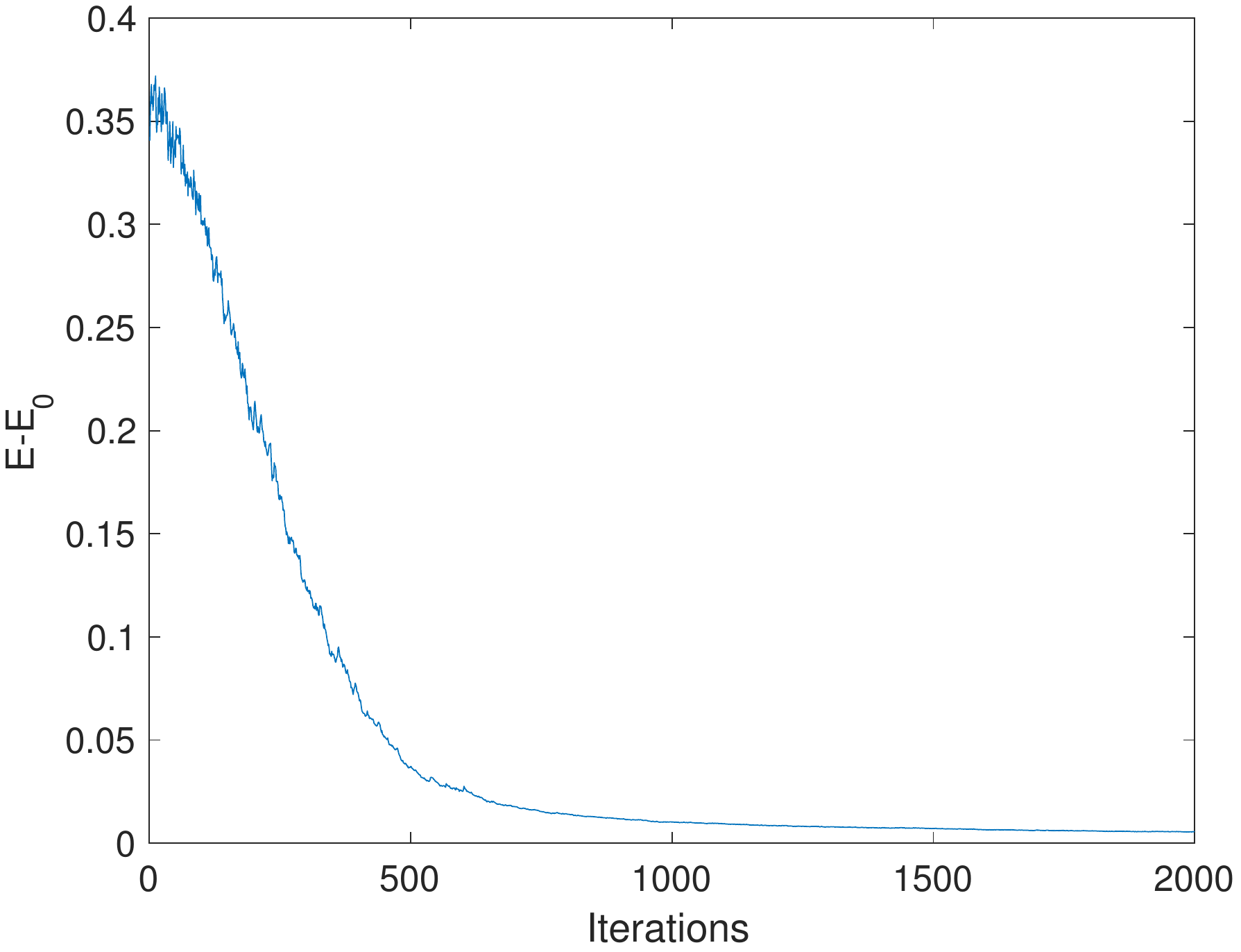}
  \caption{Energy minimization, 100 samples averaged. \label{fig:spsa_averaged}}
\end{figure}
\begin{figure}[h]
  \includegraphics[width=\columnwidth]{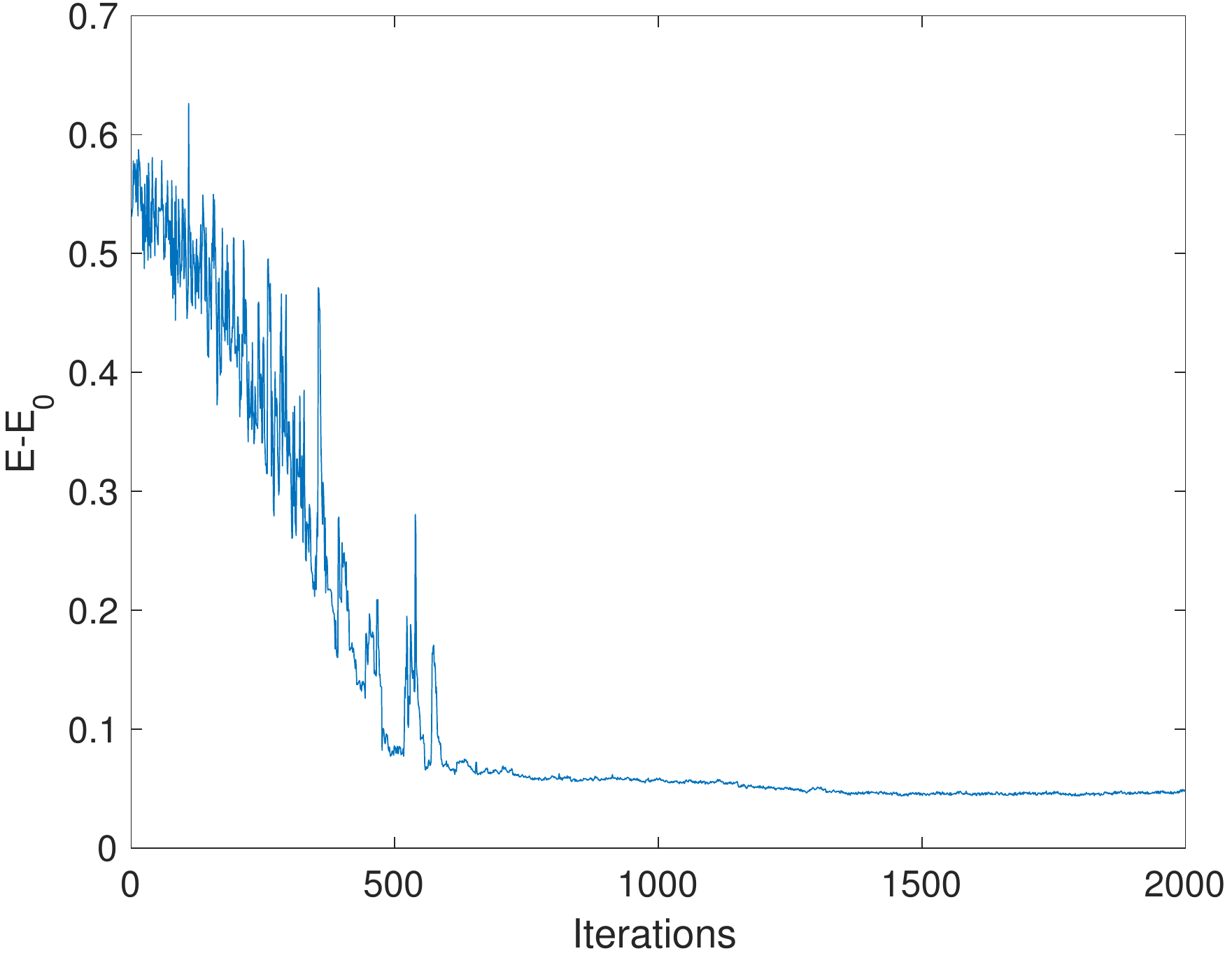}
  \caption{Energy minimization, worst case. \label{fig:spsa_worst}}
\end{figure}

One can see that, while on average the energy minimization works very well (FIG.~\ref{fig:spsa_averaged}), on rare occasions our method fails to find the true minimum (FIGS.~\ref{fig:spsa_superimposed},\ref{fig:spsa_worst}). This is a problem that also appears in the variational optimization of MERA. Fortunately, in practice even if the optimization of individual tensors fails occasionally, the global minimum can often be found by performing multiple sweeps of individual tensor optimizations \cite{Evenbly2009}. We expect our method to behave in a similar manner.

\section{Faithfulness of DMERA \label{section:faithfulness}}

The next question we address is whether DMERA can faithfully describe physical states of interest. In this section we provide evidence that the answer is yes. The basic observation is that DMERA can reproduce both scale invariant (power law) correlations and quantum critical entanglement scaling. Indeed, since any MERA can be recast as a DMERA, all the evidence that MERA can faithfully describe critical states is also evidence in favor of the suitability of DMERA for describing critical states \cite{Rizzi2010,Evenbly2011b,Corboz2010,Pineda2010}. However, this argument glosses over a key point, which is that DMERA is only practical if the required depth $D$ is not too large. Hence in this section we particularly focus on the scaling of the depth with the desired degree of approximation $\delta$. Recall that the approximation error $\delta$ is distinct from the error $\epsilon$ arising from the gates in a physical contraction of the network using a noisy quantum device. We consider both one- and two-dimensional models in the following analysis. As discussed above, DMERA can be easily generalized to more than one dimension while keeping its desirable features.

\subsection{Circuits for one dimensional models}

There are two examples in this subsection. We first consider a non-interacting fermion model which can be described by a scale invariant quantum field theory (conformal field theory or CFT) at low energies. We then turn our attention to the opposite limit of very strong interactions as realized by conformal field theories which are ``holographically dual" to gravity via the Anti de Sitter space/conformal field theory correspondence (AdS/CFT) \cite{Maldacena1998,Gubser1998,Witten1998}. In the first case we draw on recent results \cite{Haegeman2017} to demonstrate that DMERA is a good ansatz for the ground state. In the second case, we point out that ideas of holographic complexity \cite{Brown2015a,Brown2015b} similarly suggest that holographic states, suitably regulated on a lattice, can be described by a low depth DMERA.

\subsubsection{Non-interacting Dirac fermion}

Here the goal is to construct a renormalization group circuit for a non-interacting lattice fermion model which at low energies approximates a conformal field theory consisting of a free Dirac fermion. We have fermion creation operator $a_r$ on each site $r=1,...,L$ obeying the algebra
\begin{equation}
     \{ a_r , a_{r'}^\dagger \} = \delta_{r,r'}.
\end{equation}
The Hamiltonian is
\begin{equation} \label{eq:ffham}
     H = - \sum_r ( a_{r+1}^\dagger a_r + a_{r}^\dagger a_{r+1}).
\end{equation}
As is well known, this theory can be solved in momentum space to yield two Fermi points which can be combined into a lattice regulated Dirac fermion. By switching to a closely related Majorana fermion model, one can also describe the physics of the transverse Ising model spin chain using the Jordan-Wigner transformation.

For our purposes, the key physics is the following. The ground state of Eq.~\eqref{eq:ffham} is obtained by diagonalizing the corresponding single particle Hamiltonian and filling each negative energy state with a fermion, i.e., filling the ``Fermi sea". Since the model is translation invariant, the diagonalizing unitary is simply the Fourier transform, which is a non-local transformation that can map product states to states with volume law entanglement. Nevertheless, the entanglement of a region of size $\ell$ in the ground state of Eq.~\eqref{eq:ffham} is proportional to $\log \ell$, much less than volume law. Since the ground state is invariant under unitary transformations within the filled and empty single particle levels, the low degree of entanglement can be explained if a set of localized modes can be constructed which are approximately supported only on the filled single particle levels.

Recently, drawing on the theory of wavelets, Ref.~\cite{Haegeman2017} showed that such a set of localized modes can be constructed. The results build on earlier numerical explorations using wavelets, but have the virtue of providing provable error estimates. The construction is based on identifying a pair of wavelets that are so-called Hilbert transforms of each other, though the precise definitions need not concern us here. Fortunately, some time ago Selesnick showed how to construct pairs of discrete wavelet transforms that form approximate Hilbert pairs. The resulting filters depend on two parameters, call them $L_H$ and $K_H$. These filters can then be used to construct many-body unitaries that approximately implement one step of the recursive definition of DMERA in Eq.~\eqref{eq:dmerastep} using $D=2(L_H+K_H)$ layers of nearest neighbor gates.

Furthermore, Ref.~\cite{Haegeman2017} proved a theorem relating properties of the wavelet pair to the degree of approximation in the fermion DMERA circuit. Using this theorem and a numerical analysis, it was found that Selesnick's wavelets lead to an error in correlation functions that decreases exponentially with $D$. Thus to have error at most $\delta$ in local observables, it suffices to take $D \sim \log \frac{1}{\delta}$. This is essentially the best possible scaling with $\delta$, so that DMERA is an optimal approximation scheme in this case.

\subsubsection{Holographic CFTs}

Having just considered a non-interacting system, let us now turn to the opposite limit of very strongly interacting systems. In the context of AdS/CFT duality, recent developments provide a plausible conjecture for the complexity of a holographic CFT renormalization group circuit. The short version of the story is as follows: Some CFTs with many local degrees of freedom and strong interactions in $d$ dimensions are equivalent to quantum gravities in $d+1$ dimensions. The emergent dimension of the gravity theory is, in the simplest case, related to renormalization group scale in the CFT. Motivated in part by this connection to the renormalization group, it was argued that the emergence of higher dimensional gravity from CFT degrees of freedom could be understood via tensor networks like MERA~\cite{Swingle2009,Evenbly2011,Nozaki2012,Pastawski2015,Hayden2016}. Based on these tensor network ideas and other other non-trivial physical inputs, it was proposed that the complexity of the quantum state of the CFT can be calculated using the gravitational geometry \cite{Susskind2014a}.

We will not discuss the details of these calculations here, but see Refs.~\cite{Brown2015b,Lehner2016} for details. Note also that the precise definition of complexity on the CFT side is a subject of active research \cite{Jefferson2017,Chapman2017}. One simple picture is to imagine some lattice model which regulates the CFT, like the lattice fermion model above regulates the free Dirac CFT, in which the ground state has a MERA-like renormalization group circuit which prepares it from a product state. The proposal is then that the universal aspects of the complexity of this circuit can be computed on the gravity side in terms of features of the geometry. Assuming this correspondence is true, the result of the computation is that the total circuit complexity needed to prepare the ground state is proportional to the central charge $c$ times the volume of the CFT in units of the cutoff, e.g., the lattice spacing. The central charge measures the number of degrees of freedom per site; for reference, the previous fermion model has $c=1$.

To translate this complexity estimate into a depth, we need to specify how DMERA works for a system with many local degrees of freedom. Consider a one-dimensional lattice of $L$ sites where each site consists of $M$ qubits. The total Hilbert space dimension is $2^{LM}$. The typical Hamiltonian in such a system would have local few-qubit interactions between neighboring sites and arbitrary few-qubit interactions within a site. The definition of DMERA is similar to the case where $M=1$, except that now each layer of a DMERA consists of gates acting between neighboring sites and gates acting within a site. The initial product state on each interleaved site at a given layer also now consists of a product state of the $M$ qubits on that site. As a simple example, $M$ non-interacting copies of the free fermion model considered above would have central charge $c=M$ and identical depth to the $M=1$ case, i.e., the circuit for $M$ copies is simply the $M$-fold tensor product of the circuit for one copy. However, the complexity of the circuit, meaning the total number of gates, scales linearly with $c=M$. Furthermore, the total number of gates in a DMERA is proportional to the number of lattice sites.

If we now hypothesize that a holographic state with central charge $c$ can arise from a one-dimensional lattice system with $O(c)$ qubits per site, then a circuit complexity of order $c$ would translate into a depth of order unity, although the scaling of $D$ with the local error $\delta$ is not yet specified. Unfortunately, there are important subtleties that render this estimate of the depth somewhat problematic. First, there is no notion of error easy visible in the holographic computation. Second, a constant depth circuit would not allow every qubit on a given site to interact with every other qubit on the same site. As we explain below, a more reasonable conjecture is that for local error $\delta$, the necessary depth obeys $D = O\left( c \log \frac{1}{\delta} \right)$, although in some cases of shallower circuit may suffice.

This estimate arises from the holographic complexity results combined with a simple model of the holographic renormalization group circuit. To understand the model, note that the complexity measures considered in holographic models count gates near the identity as having have less complexity than gates far from the identity. This subtlety is important because at each scale every qubit on a site needs to interact with every other qubit on the same site, but the interaction is often weak, say, of order $1/c$. Hence while the total number of gates may be be of order $c^2$ (depth times number of qubits), the actual complexity is still of order $c$ since most of the gates are close to the identity.

As a simple model, imagine that each site consists of $c=M$ qubits, that every qubit interacts with every other qubit with strength $1/M$, and that DMERA needs to approximate evolution under this interaction for unit time. The evolution can be approximated by $M$ layers of two-qubit gates where each layer consists of $M/2$ gates acting on pairs of qubits and where each gate differs from the identity by an amount of order $1/M$. This way every qubit interacts with every other, but the total complexity, obtained by taking the number of gates and multiplying by the strength of each gate, is roughly proportional to $M$.

A more realistic model is provided by the so-called D1-D5 system. This is a construction within string theory consisting of a set of intersecting 1+1d and 5+1d Dirichlet branes. This setup has featured in a number of seminal results, including the first microscopic calculation of black hole entropy within string theory \cite{Strominger1996,Callan1996}. One interesting feature of this model is that it can be described, via the AdS/CFT correspondence, by a 1+1d CFT. Furthermore, this CFT has a so-called moduli space, which for our purposes here means that the theory comes with a set of continuous parameters that can be changed without breaking conformal invariance. Interestingly, in one corner of the moduli space, the CFT is described by a symmetrized version of the tensor product of many copies of the free Dirac fermion theory with many copies of a compact scalar theory (see, e.g., Ref.~\cite{Avery2010} for a recent discussion).

In fact, the compact scalar theory at a particular compactifcation radius is equivalent, via a non-local field redefinition, to the free Dirac theory. This observation suggests that the compact scalar theory has a DMERA representation with similar depth requirements as in the free Dirac theory. Furthermore, the symmetrization operation should not substantially increase the depth, at least in the limit of many copies. Thus we claim it is reasonable to conjecture (1) that the D1-D5 system in this quasi-free limit (the ``orbifold point"), suitably regulated, has a DMERA representation and (2) that this DMERA representation has complexity and depth comparable to $c$ copies of the Dirac fermion theory. Of course, this is still only in the quasi-free limit, but if the depth and complexity are smooth functions of the moduli, as might be expected if no phase transition is encountered, then we may also conjecture that the D1-D5 CFT itself, in the gravity limit, has complexity of order $c$ as predicted by the holographic calculations. It will be very interesting to further substantiate this chain of reasoning in future work.

We want to make one final comment regarding the scaling of error with depth. The fact that many of the gates in the holographic model should be close to the identity may be important in practical implementations. We assumed above that $\epsilon$ provided a uniform estimate of the error per gate, so that a depth scaling with the central charge implies the possibility of larger error. However, in most practical implementations gates near the identity are easier to implement and are less sensitive to noise, so it may be that the naively higher depth in the holographic case is actually not as detrimental as in the worst case of the same depth with arbitrary gates.

\subsection{Circuits for two dimensional models}

In the case of two-dimensional models, the basic idea is the same but the geometry of the DMERA circuit is more complex (and there are more arbitrary choices of architecture). As a simplified picture, we may imagine a square lattice of $L^2$ qubits. The recursive definition of DMERA in Eq.~\eqref{eq:dmerastep} should now relate states defined on $L^2$ qubits and $(L/2)^2$ qubits and the local unitaries should be arranged in some pattern on the square lattice. By counting degrees of freedom, we see that the number of ancillary qubits needed to go from size $L/2$ to size $L$ is $3L^2/4$. As in the one dimensional case, we consider several class of models starting from free particle models. Unlike in one dimension, there are now interesting two-dimensional topological states which can be represented using a DMERA circuit.

\subsubsection{Non-interacting Dirac point}

The simplest two-dimensional analog of the one-dimensional fermion lattice model considered above is a system with a Dirac point in its energy spectrum, e.g., as arises in a honeycomb lattice model with nearest neighbor hopping. Ref.~\cite{Haegeman2017} did consider a two-dimensional model, but instead of a Dirac point it had an entire Fermi surface, i.e., a codimension one locus of zero energy states. Nevertheless, it was shown that a generalization of DMERA equipped with a branching structure was capable representing the ground state to high accuracy. We defer the analysis of branching DMERA to a future work, but, in view of the fact that the Fermi surface state has both higher entanglement and longer-ranged correlations as compared to a Dirac point, it is reasonable to conjecture that a Dirac point would have a two-dimensional DMERA representation with depth $D \sim \log \frac{1}{\delta}$.

\subsubsection{Non-interacting Chern insulator}

We can also consider non-interacting fermion models which realize interesting topological states of matter. For example, one can construct square lattice fermion models which realize a gapped phase of matter known as a Chern insulator which hosts protected chiral edge states. In Ref.~\cite{Swingle2014} it was shown that this Chern insulator model has an adiabatic expansion property that can be used to construct a DMERA. Given the Hamiltonian $H_L$ of the Chern insulator on size $L$ and the Hamiltonian $H_{L/2}$ of the Chern insulator on size $L/2$ (plus decoupled degrees of freedom), there exists a continuous family of Hamiltonians $H(\eta)$ such that (1) $H(\eta)$ is uniformly gapped for all $\eta$, (2) $H(0)=H_{L/2}$ and $H(1)=H_{L}$, and (3) $H(\eta)$ is strictly local.

To idea to construct a DMERA is to iterate a procedure where we approximate the ground state on size $L$ by adiabatic evolution with the Hamiltonian $H(\eta)$ starting from the ground state on size $L/2$. By making the adiabatic evolution sufficiently slow, we can approximate the ground state to any desired accuracy. Furthermore, by Trotterizing the resulting adiabatic evolution, one can argue on general grounds that the requisite circuit depth scales at most as $D \sim \left( \log \frac{1}{\epsilon}\right)^q$ for some constant $q$ close to one \cite{Swingle2014}.

\subsubsection{Interacting topological states}

To address interacting topological states in two dimensions, we can follow several approaches. One is to note that for a large class of topological states, namely those described by quantum double models or string net models~\cite{Levin2005}, there are exactly solvable lattice models which have exact DMERA representations for some fixed depth $D$~\cite{Konig2009,Gu2009}. Furthermore, given a lattice models in the same phase as one of these exactly solvable points, one can construct its ground state from the corresponding exactly solvable DMERA circuit up to an adiabatic evolution for a fixed time. Since a finite time adiabatic evolution can be approximately Trotterized to a finite depth circuit, we again conclude that this large class of topological states has a good approximate DMERA representation.

The other approach starts from the observation that two-dimensional topological states are expected to have an adiabatic expansion property similar to the Chern insulator discussed above. This property is expected to apply to both non-chiral and chiral topological states \cite{Swingle2014}. Following the arguments in the Chern insulator, we can similarly conclude that an approximate DMERA representation with modest depth exists. Note that to fit the precise definition of a DMERA, it may be necessary to approximate some multi-qubit gates arising in the adiabatic evolution with two-qubit gates, but this is always possible with only an extra constant factor in the depth.

\subsubsection{Holographic CFTs}

The holographic proposal for computing CFT complexity works in any dimension \cite{Brown2015b}, so we may simply repeat the one-dimensional discussion with only minor medications. In particular, the holographic complexity is still proportional to the analog of the central charge. As before, this is consistent with a depth of order the central charge because requiring that all local degrees of freedom interact at each layer forces a blow up of the depth. Still, we conclude as in the one-dimensional case most of the gates can be taken close to the identity.

\subsection{Comments}

All of the evidence presented in this section is consistent with the broad conclusion that local observables can be obtained with a precision $\delta$ which decays exponentially (or nearly so) with the depth $D$. This conclusion indicates that DMERA should be broadly useful for describing many kinds of states in a variety of dimensions. In particular, states of matter that, in some limit, are described by Lorentz invariant quantum field theories, are all expected to fall within this class. Furthemore, none of our arguments were really specific to one or two dimensions, and we expect that three-dimensional models can also be studied.

\section{Noise resilience \label{section:noise_resilience}}

So far we have introduced DMERA and studied how well it can approximate physical states of interest. The results were encouraging: We have observed that, under plausible physical assumptions, a DMERA with moderate depth per scale can reliably approximate ground state correlation functions of many models of physical interest. For practical applications on near term noisy quantum devices, it is important to understand how this picture changes in the presence of noise.

A conservative way to estimate the effect of noise is to replace each of the individual two-qubit gates, denoted by $\mathcal{u}_j$ and understood as a completely positive trave preserving (CPTP) map, by a noisy operation $\tilde{\mathcal{u}}_j$ acting on the same qubits such that
\begin{equation}
  \|\mathcal{u}_j - \tilde{\mathcal{u}}_j \|_{\diamond} \leq \epsilon,
\end{equation}
where $\|\cdot \|_{\diamond}$ is the completely bounded norm and $\epsilon$ is the noise rate. Note that since $\mathcal{u}_j$ and $\tilde{\mathcal{u}}_j$ are CPTP maps, their dual maps $\mathcal{u}^\dagger_j$ and $\tilde{\mathcal{u}}_j^\dagger$ are unital. This has the important implication that
\begin{equation}
  \mathcal{u}_j^{\dagger}(\hat{O}) = \tilde{\mathcal{u}}_j^{\dagger}(\hat{O}) = \hat{O}
\end{equation}
for any operator $O$ whose support is disjoint with that of $\mathcal{u}_j$. This is why we can ignore, even in the presence of errors, circuit elements that lie outside the past causal cone of a given observable.

Given this noise model, what is the effect on operator expectation values? Because the quantum circuit at each layer preserves locality, an observable $\hat{O}$ supported on region of bounded size can expand to a ball of radius at most $O(D)$. By the unital property discussed above, this means that in $d$ spatial dimensions the coarse-grained operator $\Phi_i(\hat{O})$ involves at most $O(D^{d})$ circuit elements within each layer. Since $\Phi_i$ is composed of $D$ layers of circuit elements, the total number of circuit elements involved in $\Phi_i(\hat{O})$ is bounded by $O(D^{d+1})$. Therefore, we conclude that
\begin{equation}
  \|\Phi_i(\hat{O}) - \tilde{\Phi}_i(\hat{O}) \| \leq O(\epsilon D^{d+1}),
\end{equation}
where $\tilde{\Phi}$ is constructed by replacing every $\mathcal{u}_j$ to $\tilde{\mathcal{u}}_j$ and every $\ket{0}$ by a state which is close to $\ket{0}$ up to an error $\epsilon$.

The next question is how this error compounds as many $\tilde{\Phi}_i$ are applied in sequence. Naively, one might expect that noise would accumulate proportional to the number of scales of DMERA, i.e., the number of $\tilde{\Phi}_i$ applied. This would imply that the total error in local observables would grow logarithmically with the simulated system size, but this naive expectation turns out to be an overestimation. We now explain why.

Consider a translationally invariant DMERA as a variational ansatz to approximate the ground state of a translationally invariant quantum many-body system. As is the case for any variational ansatz, the correct figure of merit in the thermodynamic limit is energy per site. Without loss of generality, consider a local term in the Hamiltonian, $h$. Depending on its location, its expectation value in the DMERA state $\rho = |\psi_n \rangle \langle \psi_n|$ can be formally expressed as
\begin{equation}
\langle h \rangle = \Tr(\rho h) =  \Phi_{0} \circ \cdots \circ \Phi_{n-1}(h),
\end{equation}
where, in an abuse of notation, $\Phi_i$, $i=0,\cdots, n-1$ is the restriction of the DMERA transfer operator introduced in Section~\ref{section:proposal} to the causal cone of $h$.

We would like to bound the effect of noise for such expectation values. Consider first the total energy of the system. Let $E_0$ be the total energy per site of a noiseless DMERA and let $\tilde{E}_0$ be the total energy per site of a noisy DMERA, one whose state preparation, gates, and measurements are perturbed by a strength $\epsilon$. A naive upper bound, which is based on counting the number of perturbed locations, would scale as $|E_0 - \tilde{E}_0| < O( n \epsilon D^{d+1})$, which diverges with $n$, the logarithm of the number of qubits. This argument assumes that all $n$ scales in the DMERA are non-trivial. Away from criticality, the number $s$ of scales in a DMERA can be fixed to a system size independent value related to the correlation length, so that the analogous bound is $|E_0 - \tilde{E}_0| < O( s \epsilon D^{d+1})$. Still, at criticality we expect $s = n$ so that the error depends logarithmically on the total system size.

However, all known unitary critical theories possess a generic property that leads to a much better bound. This is the fact that the scaling dimensions of non-identity operators are positive. For example, in a unitary conformal field theory there is a dimension and spin dependent lower bound on the scaling dimension of an operator. In the context of DMERA, this fact manifests as follows. The superoperator $\Phi_i$ is unital and hence the identity operator is a eigenoperator of $\Phi_i$ with eigenvalue one. If $\lambda_i$ is the eigenvalue of the superoperator with the second largest modulus, then strict positivity of scaling dimensions implies that $|\lambda_i|<1$.

Assuming that $|\lambda_i| \leq \lambda<1$ for some $\lambda$, we can derive a bound as follows. We work in $d=1$ dimensions and use the following conventions: The number of scales in a DMERA is $s$. A symbol with a tilde sign represents the noisy counterpart of the symbol without the tilde. In particular, $\tilde{\Phi}_i$ is the noisy counterpart of $\Phi_i$ from Section~\ref{section:proposal} restricted to the causal cone of $h$. It includes all the imperfections associated with the gates and the initial states $|0\rangle$. We assume that $\|h\| \leq 1$, $\|h - \tilde{h} \|\leq \epsilon$, $\|\rho - \tilde{\rho} \|_1\leq \epsilon$, and $\|\Phi_i - \tilde{\Phi}_i \|_{\diamond}\leq O(\epsilon D^2)$, where $\|\cdots \|$ is the operator norm, $\| \cdots \|_1$ is the trace norm, and $\| \cdots \|_{\diamond}$ is the diamond norm. The factor of $D^2$ accounts for the fact that there are $\mathcal{O}(D^2)$ circuit elements appearing in the definition of $\Phi_i$. Lastly, for $m\geq n$ let $T_{[n,m]} = \Phi_n \circ \Phi_{n+1}\circ \cdots \circ \Phi_m$ and $\tilde{T}_{[n,m]} = \tilde{\Phi}_n \circ \tilde{\Phi}_{n+1} \circ \cdots \circ \tilde{\Phi}_m$; if $m<n$ set $T_{[n,m]}$ and $\tilde{T}_{[n,m]}$ to the identity map.

Recall the so-called telescopic decomposition:
\begin{equation}
\tilde{T}_{[0,s-1]} - T_{[0,s-1]} = \sum_{k=0}^{s-1} \tilde{T}_{[0,k-1]} \circ \left( \tilde{\Phi}_{k} - \Phi_{k}\right) \circ T_{[k+1,s-1]}.
\end{equation}
Also, since both $T_{[0,s-1]}$ and $\tilde{T}_{[0,s-1]}$ are unital,
\begin{equation}
  \tilde{T}_{[0,s-1]}(h) - T_{[0,s-1]}(h) = \tilde{T}_{[0,s-1]}(\bar{h}) - T_{[0,s-1]}(\bar{h}),
\end{equation}
where $\bar{h}$ is the nonunital part of $h$. Specifically, $\bar{h} = h - \Tr(h)\frac{I}{\mathcal{D}}$, where $\mathcal{D}$ is the dimension of the underlying Hilbert space. Using this decomposition,
\begin{equation}
  \begin{aligned}
    &\|\tilde{T}_{[0,s-1]}(\bar{h})- T_{[0,s-1]}(\bar{h}) \| \\
    &\leq \sum_{k=0}^{s-1}\|\tilde{T}_{[0,k-1]}\circ (\tilde{\Phi}_k - \Phi_k)\circ T_{[k+1,s-1]}(\bar{h}) \| \\
    &\leq \sum_{k=0}^{s-1}\|(\tilde{\Phi}_k - \Phi_k)\circ T_{[k+1,s-1]}(\bar{h})\| \\
    &\leq O(\epsilon D^2) \sum_{k=0}^{s-1}\|T_{[k+1,s-1]}(\bar{h}) \| \\
    &\leq O(\epsilon D^2)\sum_{k=0}^{s-1}  |\lambda|^{s-1-k} \\
    &\leq \frac{O(\epsilon D^2)}{1-|\lambda|}.
  \end{aligned}
  \label{eq:error_bound}
\end{equation}
From the second to the third line, we used the fact that $\tilde{T}_{[0,k-1]}$ is norm-nonincreasing. Next, we used the fact $\|\Phi_k - \tilde{\Phi}_k \|\leq O(\epsilon D^2)$. Then we used the fact that $\bar{h}$ is the nonunital part of $h$: provided that $|\lambda|<1$, its norm strictly contracts.

Note that the bound is independent of $s$. Therefore even at criticality, the expectation values of local observables are perturbed at most by $O(\epsilon D^2)$, independent of the system size.
An immediate consequence is that, provided the noiseless DMERA ansatz can well-approximate the ground state energy, a noisy ansatz can approximate the ground state energy up to an additional $O(\epsilon D^2)$ error. Therefore, if a small value of $D$ suffices for a good approximation to the ground state, and if the energy minimization discussed in the previous section finds the minimum, a moderately noisy quantum computer ought to be able to reliably estimate the total energy.

This bound formally shows that the error in local expectation values is independent of the system size, despite the fact that the depth of the circuit scales with the system size. In order to assess how stable the circuit is in practice, we have performed a number of numerical experiments. We have computed the reduced density matrix of two neighboring sites in the middle of the chain for the noiseless circuit, and compared it with the reduced density matrix obtained from a noisy circuit. Each of the two-qubit gates are chsoen uniformly random from $SU(4)$. In the noisy circuit, each gates are followed by a depolarizing noise with error rate $p=0.001$ on qubits the gate acts on. The trace distance between these densiey matrices are numerically computed. The result is summarized in Table \ref{table:numerics}.
\begin{table}[h]
  \begin{tabular}{c|c|c|c|c|c}
    $D$ & Errors& Average & Std & Min & Max \\
    \hline
    2 & 624  & $7.2 \times 10^{-3}$ &  $1.2 \times 10^{-3}$ & $5.8 \times 10^{-3}$ & $10.7 \times 10^{-3}$ \\
    \hline
    3 & 1656 & $7.3 \times 10^{-3}$ & $1.0 \times 10^{-3}$ & $5.1\times 10^{-3} $ & $9.9 \times 10^{-3}$\\
    \hline
    4 & 3048 & $1.4 \times 10^{-2}$ & $1.6 \times 10^{-3}$ & $1.1 \times 10^{-2}$ & $1.8 \times 10^{-2}$ \\
    \hline
    5 & 4680 & $1.9 \times 10^{-2}$ & $1.5 \times 10^{-3}$ & $1.6 \times 10^{-2}$ & $2.3 \times 10^{-2}$
  \end{tabular}
  \caption{Trace distance between the reduced density matrix over the $255$th and the $256$th qubit over a spin chain consisting of $512$ qubits. Errors represent the number of locations in the circuit where a depolarizing noise of error $p=0.001$ has occurred. Average, Min, and Max denote the average, minimum, and the maximum value within the sample. The total sample size is for each $D$. Also, Std represents the standard deviation.\label{table:numerics}}
\end{table}

The outcome of our numerical experiment is encouraging. Based on the number of errors and the error rate, one might have expected the trace distance to be so large that the two reduced density matrices are almost perfectly distinguishable. However, as one can see in Table \ref{table:numerics}, the reported trace distance, even in the worst case, is several orders of magnitude below this expectation. It is importnat to note that the trace distance upper bounds the error on the estimated expectation values for normalized operators; the error for physical observables will be lower in general.

\subsection{Comparing with classical numerical methods}

Suppose an experimentalist has variationally found a low-energy state and has estimated its energy density. To what extent can we trust this result? One thing that is clear is that the experimentalist must have done a admirable job of characterizing the gates of the quantum computer. Then we can expect the experimentalist to have a reliable estimate of $\epsilon$. The second largest eigenvalue of the transfer operator, $\lambda$, is an experimentally measurable quantity; simply apply an operator $\Phi_{n}$ multiple times and measure the expectation value of all possible operators. The slowest rate at which expectation values of these observables equilibrate determine $\lambda$. From these estimates, we can upper bound the deviation in the expectation value from Eq.~\eqref{eq:error_bound}. Then a variational upper bound to the energy will be the estimated energy with an addition of this error, as well as the statistical error.

In practice, estimating $\lambda$ may be unpractical. We shall revisit this issue in Section \ref{section:contraction} and provide a method that obviates the need to measure $\lambda$.

\subsection{Extracting physical information}

The preceding analysis implies that the energy per site can reach a precision of $O(\epsilon)$, provided that $D$ is small and the ansatz approximates the ground state well. Because $\epsilon$ is a constant, as opposed to a number that decays polynomially in the system size, the stability bound implies that a noisy quantum computer can prepare \emph{some} state at low but finite energy density.

This may seem underwhelming; we do not have an approximation of a ground state, but rather some low energy state. However, there are successful variational methods such as iDMRG~\cite{Vidal2007} and iPEPS~\cite{Jordan2008} which directly target the thermodynamic limit. The correct figure of merit for these methods is the energy per site, and the same criterion should be applied to our method. After all, the only difference between our method and the existing variational calculations is the employed computational device. If our method reports an energy per site lower than any of the existing calculations on some model, that will be a strong evidence for the usefulness of our proposal.

\section{Contraction Algorithm \label{section:contraction}}

We introduce an algorithm for contracting a DMERA on a quantum computer. The algorithm outputs, for every ball of radius $O(D)$, a sequence of gates that prepares the reduced density matrix of the ball. Importantly, the energy obtained from these reduced density matrices is variational; see Eq.~\eqref{eq:variational}.

The main idea is to assign the physical qubits to the qubits that appear in the circuit diagram in such a way that a consistent quantum state is obtained. The distinction between these two types of qubits is important. As such, we shall refer to the qubits appearing in the circuit diagram as \emph{circuit qubits}, and we shall refer to the qubits that are present in the quantum computer as \emph{physical qubits}.

For concreteness, we use the following set of conventions. First, we interpret preparation as a single-qubit gate, as opposed to a single-qubit state. Specifically, it is a CPTP map that replaces the input state with some fixed state. In other words, preparation can be thought as a ``reset'' gate. Second, we write down the state prepared by the DMERA circuit as
\begin{equation}
  \rho_{\text{DMERA}} = \mc_n \circ \cdots \circ \mc_1(\sigma),
\end{equation}
where $\sigma$ is a maximally mixed state over all the circuit qubits that appear in the DMERA circuit, and $\mc_i$ is the action of all the gates that appear on the $i$th scale of the DMERA circuit. This can be further decomposed into
\begin{equation}
  \mc_{i} = \mc_{i,D} \circ \cdots \mc_{i,1} \circ \mc_{i,0},
\end{equation}
where $\mc_{i,0}$ consists only of the preparation gates and $\mc_{i,j}$ corresponds to the action of each layer of the circuit.

We consider a system on a $d$-dimensional hypercubic lattice consisting of $(2^n(2\ell_0+1))^d$ qubits, where $\ell_0$ is a nonnegative integer. Our DMERA prepares a state over these qubits; note that this slightly generalizes our definition of DMERA in Section~\ref{section:proposal}. The $i$th layer from the top consists of $2^{di}(2\ell_0+1)^d$ qubits. In particular, $(2^d-1) \times 2^{d(i-1)}(2\ell_0+1)^d$ new qubits are introduced in the $i$th layer.

The rest of this section is divided into three parts. First, we propose a particular assignment of physical qubits to the qubits that appear in the circuit diagram. Second, based on this assignment, we propose an algorithm that outputs a sequence of gates for every ball of radius $O(D)$; sequential application of these gates prepares the reduced density matrix of the ball. Third, these reduced density matrices are shown to be consistent with some global state.

\subsection{Qubit assignment}

Recall that there are $(2^n(2\ell_0+1))^d$ qubits that appear in the description of DMERA. For each of these qubits, we assign a physical qubit. What is important about this assignment is that for every past causal cone of a local observable, and within every layer of these past causal cones, distinct physical qubit is assigned to every circuit qubit. The existence of such assignment is a nontrivial fact, and it will play an important role in the analysis of the algorithm. Among other things, this assignment ensures that the total energy estimated from the method is variational.

In order to make the analysis clean, we make the following set of choices. First, we assume that the qubits are supported on a square lattice of size $\underbrace{2^n(2\ell_0+1) \times \cdots \times 2^n(2\ell_0+1)}_{d}$, where $\ell_0$ is a nonnegative integer. In the first (top) layer of DMERA, there are $(2\ell_0+1)^d$ qubits. As we move from the $i$th layer to the $(i+1)$th layer, the number of qubits is increased by a factor of $2^d$. Out of these, $1/2^d$ of them descend from the $i$th layer, while the rest are initialized to some fixed state and introduced to the layer. After a local depth-$D$ circuit is applied, the process repeats.

Here is the assignment. For the circuit qubits in the top layer, assign different physical qubits. Let us denote this assignment as $a_{\vec{x}}^{(1)}$, where $\vec{x}$ is a lattice vector $\vec{x}=(x_1,\ldots,x_d)$ with $x_j \in \{0,\ldots,2\ell_0\}$. Given an assignment on the $i$th layer, we recursively define the assignment on the $i+1$th layer as follows. Let $a_{\vec{x}}^{(i)}$ be the assignment for the $i$th layer. Again $\vec{x}$ is a lattice vector, but its range is different: $x_j \in \{0,\ldots, 2^{i-1}(2\ell_0+1) -1\}$. Then
\begin{equation}
  \begin{aligned}
    a_{2\vec{x} + \vec{u}}^{(i+1)} &= a_{F_{\vec{u}}(\vec{x})}^{(i)}, \label{eq:recursive_assignment}
  \end{aligned}
\end{equation}
where $\vec{u}$ is a vector whose entries are either $0$ or $1$ and $F_{\vec{u}}^{(i)}$ is defined as
\begin{equation}
  F_{\vec{u}} (\vec{x})= (\vec{x} + (\ell_0+1)\vec{u}) \mod 2\ell_0+1,
\end{equation}
with the $\mod 2\ell_0+1$ applied to the entries of the vector. This assignment relies on the fact that any vector $\vec{x}$ on the $(i+1)$the layer can be written as $2\vec{x}'+\vec{u}$ where $\vec{x}'$ is a vector on the $i$th layer. In particular, the choice of $\vec{x}'$ and $\vec{u}$ is unique. Here $\vec{u}$ specifies whether $\vec{x}$ is located in the even or odd lattice for each coordinate, and $\vec{x}'$ specifies the coordinate within each of these $2^d$ classes.

Because the assignment at layer $i+1$ is defined in terms of the assignment at level $i$, a total of $(2\ell_0+1)^d$ physical qubits are assigned. An example for $d=1$ is drawn in FIG.~\ref{fig:assignment_recursive}. We now demonstrate some properties of this assignment which will be useful later.
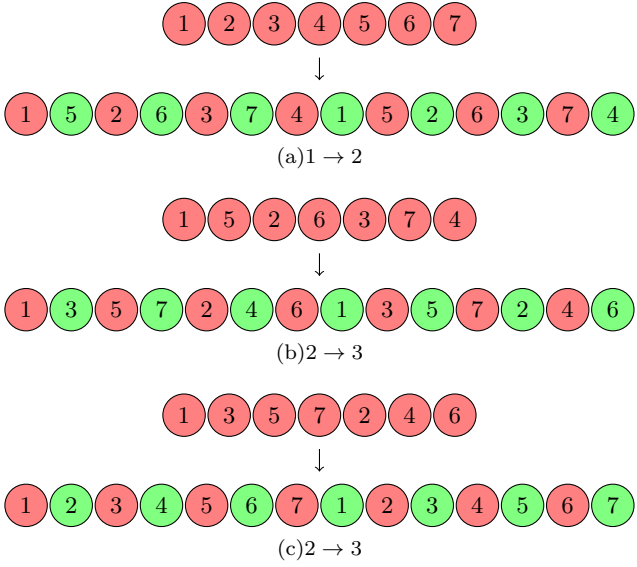
\begin{figure}[h]
  \subfigure[$1\to 2$]{
  \begin{tikzpicture}[scale=0.6]
    \foreach \x in {1,...,7}
    {
      \draw[] (\x-3.5,0) node[circle,fill=red!50,draw=black] {\small $\x$};
    }
    \draw[->] (0.5,-0.75) -- ++ (0,-0.5);
    \foreach \x in {1,...,7}
    {
      \draw[] (2*\x-8,-2) node[circle,fill=red!50,draw=black] {\small $\x$};
    }
    \foreach \x in {5,...,7}
    {
      \draw[]  (2*\x-15,-2) node[circle,fill=green!50,draw=black] {\small $\x$};
    }
    \foreach \x in {1,...,4}
    {
      \draw[]  (2*\x -1,-2) node[circle,fill=green!50,draw=black] {\small $\x$};
    }
  \end{tikzpicture}
}
  \subfigure[$2\to 3$]{
  \begin{tikzpicture}[scale=0.6]
    \draw[] (1-3.5,0) node[circle,fill=red!50,draw=black] {\small $1$};
    \draw[] (2-3.5,0) node[circle,fill=red!50,draw=black] {\small $5$};
    \draw[] (3-3.5,0) node[circle,fill=red!50,draw=black] {\small $2$};
    \draw[] (4-3.5,0) node[circle,fill=red!50,draw=black] {\small $6$};
    \draw[] (5-3.5,0) node[circle,fill=red!50,draw=black] {\small $3$};
    \draw[] (6-3.5,0) node[circle,fill=red!50,draw=black] {\small $7$};
    \draw[] (7-3.5,0) node[circle,fill=red!50,draw=black] {\small $4$};

    \draw[->] (0.5,-0.75) -- ++ (0,-0.5);

    \draw[] (2*1-8,-2) node[circle,fill=red!50,draw=black] {\small $1$};
    \draw[] (2*2-8,-2) node[circle,fill=red!50,draw=black] {\small $5$};
    \draw[] (2*3-8,-2) node[circle,fill=red!50,draw=black] {\small $2$};
    \draw[] (2*4-8,-2) node[circle,fill=red!50,draw=black] {\small $6$};
    \draw[] (2*5-8,-2) node[circle,fill=red!50,draw=black] {\small $3$};
    \draw[] (2*6-8,-2) node[circle,fill=red!50,draw=black] {\small $7$};
    \draw[] (2*7-8,-2) node[circle,fill=red!50,draw=black] {\small $4$};

    \draw[] (2*1-7,-2) node[circle,fill=green!50,draw=black] {\small $3$};
    \draw[] (2*2-7,-2) node[circle,fill=green!50,draw=black] {\small $7$};
    \draw[] (2*3-7,-2) node[circle,fill=green!50,draw=black] {\small $4$};
    \draw[] (2*4-7,-2) node[circle,fill=green!50,draw=black] {\small $1$};
    \draw[] (2*5-7,-2) node[circle,fill=green!50,draw=black] {\small $5$};
    \draw[] (2*6-7,-2) node[circle,fill=green!50,draw=black] {\small $2$};
    \draw[] (2*7-7,-2) node[circle,fill=green!50,draw=black] {\small $6$};
  \end{tikzpicture}
}
\subfigure[$2\to 3$]{
  \begin{tikzpicture}[scale=0.6]
    \draw[] (1-3.5,0) node[circle,fill=red!50,draw=black] {\small $1$};
    \draw[] (2-3.5,0) node[circle,fill=red!50,draw=black] {\small $3$};
    \draw[] (3-3.5,0) node[circle,fill=red!50,draw=black] {\small $5$};
    \draw[] (4-3.5,0) node[circle,fill=red!50,draw=black] {\small $7$};
    \draw[] (5-3.5,0) node[circle,fill=red!50,draw=black] {\small $2$};
    \draw[] (6-3.5,0) node[circle,fill=red!50,draw=black] {\small $4$};
    \draw[] (7-3.5,0) node[circle,fill=red!50,draw=black] {\small $6$};

    \draw[->] (0.5,-0.75) -- ++ (0,-0.5);

    \draw[] (2*1-8,-2) node[circle,fill=red!50,draw=black] {\small $1$};
    \draw[] (2*2-8,-2) node[circle,fill=red!50,draw=black] {\small $3$};
    \draw[] (2*3-8,-2) node[circle,fill=red!50,draw=black] {\small $5$};
    \draw[] (2*4-8,-2) node[circle,fill=red!50,draw=black] {\small $7$};
    \draw[] (2*5-8,-2) node[circle,fill=red!50,draw=black] {\small $2$};
    \draw[] (2*6-8,-2) node[circle,fill=red!50,draw=black] {\small $4$};
    \draw[] (2*7-8,-2) node[circle,fill=red!50,draw=black] {\small $6$};

    \draw[] (2*1-7,-2) node[circle,fill=green!50,draw=black] {\small $2$};
    \draw[] (2*2-7,-2) node[circle,fill=green!50,draw=black] {\small $4$};
    \draw[] (2*3-7,-2) node[circle,fill=green!50,draw=black] {\small $6$};
    \draw[] (2*4-7,-2) node[circle,fill=green!50,draw=black] {\small $1$};
    \draw[] (2*5-7,-2) node[circle,fill=green!50,draw=black] {\small $3$};
    \draw[] (2*6-7,-2) node[circle,fill=green!50,draw=black] {\small $5$};
    \draw[] (2*7-7,-2) node[circle,fill=green!50,draw=black] {\small $7$};
  \end{tikzpicture}
  }
  \caption{Recursive assignment for a $d=1$ DMERA, where $\ell_0$ is chosen to be $3$. This assignment, applied to a DMERA of depth $D=2$, is sufficient for a set of nearest-neighbor observables. Physical qubits assigned at level $i$ (red) are interleaved with the physical qubits assigned at level $i+1$(green).\label{fig:assignment_recursive}}
\end{figure}

First, we show that the assignments are invariant under a shift by $(2\ell_0+1)$ in any direction.
\begin{lem}
  For any unit vector $\hat{x}_n$, $n\in \{1,\ldots,d \}$,
  \begin{equation}
    a_{\vec{x}}^{(i)} = a_{\vec{x}+(2\ell_0+1)\hat{x}_n}^{(i)}.
  \end{equation}
  \label{lemma:shift_invariance}
\end{lem}
\begin{proof}
  Note that any vector on the $(i+1)$th layer can be expressed as $2\vec{x} + \vec{u}$, where $\vec{x}$ is a vector on the $i$th layer and $\vec{u}$ is a vector whose components are either $0$ or $1$. Let us consider two possibilities, depending on whether the $n$th component of $\vec{u}$ is $0$ or $1$. Obviously, this covers all the possibilities by definition. If $u_n =0$, then we can write $2\vec{x}+\vec{u} + (2\ell_0+1)\hat{x}_n$ as $2(\vec{x}+\ell_0 \hat{x}_n) + \vec{u} +\hat{x}_n$ and so
  \begin{equation}
    \begin{aligned}
      a_{2\vec{x}+\vec{u}}^{(i+1)} &= a_{F_{\vec{u}}(\vec{x})}^{(i)}\\
      a_{2\vec{x} + \vec{u} + (2\ell_0+1)\hat{x}_n}^{(i+1)} &= a_{F_{\vec{u}+\hat{x}_n}(\vec{x} + \ell_0\hat{x}_n) }^{(i)}.
    \end{aligned}
  \end{equation}
  Recalling that
  \begin{equation}
    F_{\vec{u}}(\vec{x}) = (\vec{x} + (\ell_0+1)\vec{u}) \mod (2\ell_0+1),
  \end{equation}
  we compute
  \begin{equation}
    \begin{aligned}
      &F_{\vec{u}+\hat{x}_n}(\vec{x} + \ell_0\hat{x}_n)) \\
      &= (\vec{x} + (2\ell_0+1)\hat{x}_n + (\ell_0+1)\vec{u}) \mod (2\ell_0+1) \\
      &= (\vec{x} + (\ell_0+1)\vec{u}) \mod (2\ell_0+1) \\
      &= F_{\vec{u}}(\vec{x}).
    \end{aligned}
  \end{equation}
  Therefore, if the $n$th component of $\vec{u}$ is $0$, then $a_{2\vec{x}+\vec{u}}^{(i+1)} = a^{(i+1)}_{2\vec{x}+\vec{u}+ (2\ell_0+1)\hat{x}_n}$.

  If $u_n$ is $1$, then by a similar argument we find
  \begin{equation}
    \begin{aligned}
      a_{2\vec{x}+\vec{u}}^{(i+1)} &= a_{F_{\vec{u}}(\vec{x})}^{(i)}\\
      a_{2\vec{x} + \vec{u} + (2\ell_0+1)\hat{x}_n}^{(i+1)} &= a_{F_{\vec{u}-\hat{x}_n}(\vec{x} + (\ell_0+1)\hat{x}_n) }^{(i)}.
    \end{aligned}
  \end{equation}
  In this case,
  \begin{equation}
    \begin{aligned}
      &F_{\vec{u}-\hat{x}_n}(\vec{x} + (\ell_0+1)\hat{x}_n) \\
      &= (\vec{x} + (\ell_0+1)\hat{x}_n + (\ell_0+1)(\vec{u}-\hat{x}_n)) \mod (2\ell_0+1)\\
      &= (\vec{x} + (\ell_0+1)\vec{u}) \mod (2\ell_0+1) \\
      &= F_{\vec{u}}(\vec{x}).
    \end{aligned}
  \end{equation}
  Hence, if the $n$th component of $\vec{u}$ is $1$, then $a_{2\vec{x} + \vec{u}}^{(i+1)} = a^{(i+1)}_{2\vec{x} + \vec{u} + (2\ell_0+1)\hat{x}_n}$.

  Therefore, for all $\vec{u}$ whose components are $0$ or $1$, and for all $\vec{x}$,
  \begin{equation}
    a_{2\vec{x} +\vec{u}}^{(i+1)} = a_{2\vec{x} + \vec{u} + (2\ell_0+1)\hat{x}_n}^{(i+1)}.
  \end{equation}
  Since any vector in the $(i+1)$th layer can be expressed as $2\vec{x} + \vec{u}$ in terms of a vector $\vec{x}$ in the $i$th layer, the claim is true for any $i\geq 2$. The $i=1$ case is trivially true.
\end{proof}

Lemma \ref{lemma:shift_invariance} shows that the qubit assignment is invariant under a shift of $(2\ell_0+1)$ in any direction. This implies that two different circuit qubits that are assigned the same physical qubit must be at least distance $2\ell_0+1$ apart from each other, as we explain below. The following lemma will be useful for elucidating the argument.
\begin{lem}
  Let $\mathcal{A}(i)$ be the number of distinct physical qubits assigned to the circuit qubits on the $i$th layer. For all $i$,  $\mathcal{A}(i)= (2\ell_0+1)^d$.
  \label{lemma:counting}
\end{lem}
\begin{proof}
  First, note that the entire qubit assignment on the $(i+1)$th layer is defined by Eq.~\eqref{eq:recursive_assignment}. Therefore, $\mathcal{A}(i+1) \leq \mathcal{A}(i)$. In the remaining part of the proof, we show that $\mathcal{A}(i) \geq (2\ell_0+1)^d$ for all $i$. The $i=1$ case is trivial. For $i>1$, note that
  \begin{equation}
    \begin{aligned}
      a_{2^{i}\vec{x}}^{(i+1)} &= a_{2^{i-1}\vec{x} \mod 2\ell_0+1}^{(i)} \\
      &= a_{2^{i-1}\vec{x}}^{(i)} \\
      &= a_{\vec{x}}^{(1)}
    \end{aligned}
  \end{equation}
  where in the first line we used Eq.~\eqref{eq:recursive_assignment} and in the second line we used Lemma \ref{lemma:shift_invariance}. The third line follows from a recursion of the same relation. Therefore, for every layer, there is a physical qubit assigned that appeared in the first layer. Since the number of assigned qubits in the first layer is $(2\ell_0+1)^d$, it follows that $\mathcal{A}(i) \geq (2\ell_0+1)^d$. Combining the two bounds, the claim follows.
\end{proof}

Now we show that different circuit qubits which are assigned the same physical qubit must be a certain distance apart.

\begin{thm}
  For any $\vec{x},\vec{y}$ on the $i$-th layer
  \begin{equation}
    a_{\vec{x}}^{(i)} = a_{\vec{y}}^{(i)}
  \end{equation}
  if and only if $\vec{x} - \vec{y} = (2\ell_0+1)\vec{z}$ for some integer-valued vector $\vec{z}$.
\end{thm}
\begin{proof}
  The ``if'' part follows trivially from Lemma \ref{lemma:shift_invariance}. For the ``only if'' part, suppose there is $\vec{x}$ and $\vec{y}$ such that $a_{\vec{x}}^{(i)} = a_{\vec{y}}^{(i)}$ yet $\vec{x} - \vec{y}\neq (2\ell_0+1)\vec{z}$ for any integer-valued vector $\vec{z}$. By applying Lemma \ref{lemma:shift_invariance} to both $\vec{x}$ and $\vec{y}$, one can conclude that there are at least two circuit qubits with the same physical qubit in a set $\{\vec{v}|v_i \in \{0,\ldots,2\ell_0 \}  \quad \forall i\in \{1,\ldots,d \}\}$. Therefore, the number of distinct physical qubits assigned to this set is strictly smaller than $(2\ell_0+1)^d$. Because the qubit assignment is invariant under a shift of $2\ell_0+1$ in any direction, the total number of physical qubits assigned to the entire $i$th layer is strictly smaller than $(2\ell_0+1)^d$. This contradicts Lemma \ref{lemma:counting}.
\end{proof}

Based on this Theorem, we conclude that, at every layer, two different circuit qubits with the same physical qubit assignment must be at least distance $2\ell_0+1$ away from each other. Therefore, in a hypercube of size $(2\ell_0+1)^d$, every circuit qubit is assigned a different physical qubit. The diameter of the past causal cone of DMERA is at most $4D-1$. Therefore, by choosing $\ell_0=2D-1$, the qubit assignment of Eq.~\eqref{eq:recursive_assignment} becomes sufficient for every observable supported on a box of size $(2D)^d$.

\subsection{Algorithm}
Let $\mathfrak{A}: \mathcal{C} \to \mathcal{P}$ be the assignment described in Eq.~\eqref{eq:recursive_assignment}. Here $\mathcal{C}$ is the set of circuit qubits and $\mathcal{P}$ is a set of physical qubits. Without loss of generality, consider an observable supported in a hypercube of side length $2D$. Let us refer to this hypercube as $C$. Then consider a hypercube $C'$ of side length $4D$ such that the centers of $C$ and $C'$ coincide. Let $C(n)=C$ and $C'(n)=C'$.

Within the bottom layer, there will be $|C'(n)|(1-1/2^d)$ prepration gates acting on $C'(n)$. Remove the support of these gates from $C'(n)$ and define the remaining circuit qubits as $C(n-1)$. Similar to the relation between $C(n)$ and $C'(n)$, let $C'(n-1)$ be a hypercube of side length $4D$ such that its center coincides with $C(n-1)$  on the $(n-1)$th scale layer. $C(i)$ and $C'(i)$ are similarly defined for all $1 \leq i\leq n$.

Recall that each scale layer consists of $D+1$ layers, wherein the first layer only consists of preparation gates and the rest consists of two-qubit gates. The algorithm is to apply these gates sequentially, layer by layer, but to the \emph{assigned physical qubits.}

Even though we did not specify which gate comes first within the layer, the order does not matter as long as the gates are within the same layer. This is because we made sure that every circuit qubits within a hypercube of side length $4D$ are all distinct. Once all these gates are applied, all that remains is to measure the observable. Without loss of generality, any observable can be decomposed into a linear combination Pauli operators. Therefore, one can measure any observable by computing this decomposition, and measuring each Pauli operator in its eigenbasis. Which qubits do we measure? For each Pauli operator associated to a circuit qubit, one can simply measure the assigned physical qubit in the eigenbasis. Again, there is no ambiguity here because every observable is contained in a hypercube of side length $4D$ and we required $\ell_0\geq 2D-1$.

\subsection{Existence of a consistent global state}

If every gate is perfect, the algorithm outputs the correct reduced density matrices of the DMERA. But what if the gates are not perfect? Whenever an experimentalist implements a gate, say $U$, she is actually applying some physical operation, say $\mc_U$. She would certainly like $\mc_U$ to be equal to $U$, but this is never going to be exactly right. The outcome will be resilient to noise if $\mc_U$ is a good approximation to $U$. However, even if they are completely unrelated, the obtained reduced density matrices are consistent with \emph{some} quantum state.

In the rest of the section, we explicitly construct a state whose reduced density matrices are exactly equal to the reduced density matrices obtained from the algorithm. This state is specified in terms of $\mc_U$. We emphasize that we make no assumption about $\mc_U$ aside from the fact that it is a CPTP map applied to the set of qubits that $U$ acts on. Specifically, every preparation is modeled by a preparation of some fixed state, which can be mixed in general. Every two-qubit gate is modeled by a CPTP map acting on the same set of qubits.

The state we construct is a noisy version of the DMERA state wherein each gates are replaced by their noisy counterparts. We will extensively use the following set of diagrams.
\begin{equation}
  \begin{aligned}
  \rho_0^{(c)} &=
  \begin{tikzpicture}[scale=0.5,baseline={([yshift=-.5ex]current bounding box.center)}]
    \draw[thick] (-1,0)--(1,0)--(0,0.8665)--cycle;
    \draw[] (0,0) -- (0,-0.8665);
    \draw[] (0,-0.43325) node[left] {$c$};
  \end{tikzpicture},
  \\
  \Tr_c &=
  \begin{tikzpicture}[scale=0.5,baseline={([yshift=-.5ex]current bounding box.center)}]
    \draw[thick] (-1,0)--(1,0)--(0,-0.8665)--cycle;
    \draw[] (0,0) -- (0,0.8665);
    \draw[] (0,0.43325) node[left] {$c$};
  \end{tikzpicture},
  \\
  \mc_U &=
  \begin{tikzpicture}[scale=0.5,baseline={([yshift=-.5ex]current bounding box.center)}]
    \draw[thick] (-1,0)--(1,0)--(1,0.8665)--(-1,0.8665)--cycle;
    \draw[] (0,0.43325) node[] {$\mc_U$};
    \draw[] (-0.75,0) -- (-0.75,-0.8665);
    \draw[] (0.75,0) -- (0.75,-0.8665);
    \draw[] (-0.75,0.8665) -- (-0.75,1.732);
    \draw[] (0.75,0.8665) -- (0.75,1.732);
    \draw[] (-0.75,-0.43325) node[left] {$c_1$};
    \draw[] (-0.75,1.29975) node[left] {$c_1$};
    \draw[] (0.75,-0.43325) node[right] {$c_2$};
    \draw[] (0.75,1.29975) node[right] {$c_2$};
  \end{tikzpicture}.
\end{aligned}
\end{equation}
where the indices $c$, $c_1$, and $c_2$ represent the circuit qubits that these gates act on.
Note the following identities:
\begin{equation}
  \begin{aligned}
    \begin{tikzpicture}[scale=0.5,baseline={([yshift=-.5ex]current bounding box.center)}]
    \draw[thick] (-1,0)--(1,0)--(0,0.8665)--cycle;
    \draw[] (0,0) -- (0,-0.8665);
    \draw[thick] (-1,-0.8665)--(1,-0.8665)--(0,-1.732)--cycle;
    \draw[] (0,-0.43325) node[left] {$c$};
  \end{tikzpicture}
  &= 1, \\
  \begin{tikzpicture}[scale=0.5,baseline={([yshift=-.5ex]current bounding box.center)}]
    \draw[thick] (-1,0)--(1,0)--(1,0.8665)--(-1,0.8665)--cycle;
    \draw[] (0,0.43325) node[] {$\mc_U$};
    \draw[] (-0.75,0) -- (-0.75,-0.8665);
    \draw[] (0.75,0) -- (0.75,-0.8665);
    \draw[] (-0.75,0.8665) -- (-0.75,1.732);
    \draw[] (0.75,0.8665) -- (0.75,1.732);
    \draw[] (-0.75,-0.43325) node[left] {$c_1$};
    \draw[] (-0.75,1.29975) node[left] {$c_1$};
    \draw[] (0.75,-0.43325) node[right] {$c_2$};
    \draw[] (0.75,1.29975) node[right] {$c_2$};
    \draw[thick] (-0.75,-0.8665) -- ++ (0.5,0) -- ++ (-0.5,-0.43325) -- ++ (-0.5,0.43325)--cycle;
    \draw[thick] (0.75,-0.8665) -- ++ (0.5,0) -- ++ (-0.5,-0.43325) -- ++ (-0.5,0.43325)--cycle;
  \end{tikzpicture}
  &=
  \begin{tikzpicture}[scale=0.5,baseline={([yshift=-.5ex]current bounding box.center)}]
    \draw[] (-0.75,0) -- (-0.75,-0.8665);
    \draw[] (0.75,0) -- (0.75,-0.8665);
    \draw[] (-0.75,-0.43325) node[left] {$c_1$};
    \draw[] (0.75,-0.43325) node[right] {$c_2$};
    \draw[thick] (-0.75,-0.8665) -- ++ (0.5,0) -- ++ (-0.5,-0.43325) -- ++ (-0.5,0.43325)--cycle;
    \draw[thick] (0.75,-0.8665) -- ++ (0.5,0) -- ++ (-0.5,-0.43325) -- ++ (-0.5,0.43325)--cycle;
  \end{tikzpicture}.
  \end{aligned}
\end{equation}
We shall also use the following short-hand notations for representing composite objects.
\begin{equation}
  \begin{aligned}
  \bigotimes_{c\in C} \rho_0^{(c)} &=
  \begin{tikzpicture}[scale=0.5,baseline={([yshift=-.5ex]current bounding box.center)}]
    \draw[thick] (-1,0)--(1,0)--(0,0.8665)--cycle;
    \draw[] (0,0) -- (0,-0.8665);
    \draw[] (0,-0.43325) node[left] {$C$};
  \end{tikzpicture},
  &\Tr_C &=
  \begin{tikzpicture}[scale=0.5,baseline={([yshift=-.5ex]current bounding box.center)}]
    \draw[thick] (-1,0)--(1,0)--(0,-0.8665)--cycle;
    \draw[] (0,0) -- (0,0.8665);
    \draw[] (0,0.43325) node[left] {$C$};
  \end{tikzpicture}.
\end{aligned}
\end{equation}
where $C$ is a set of circuit qubits. Also, the following diagram will play an important role:
\begin{equation}
  \begin{tikzpicture}[scale=0.5,baseline={([yshift=-.5ex]current bounding box.center)}]
    \draw[thick] (-1,0)--(1,0)--(1,0.8665)--(-1,0.8665)--cycle;
    \draw[] (0,0.43325) node[] {$\mc_i$};
    \draw[] (-0.75,0) -- (-0.75,-0.8665);
    \draw[] (0.75,0) -- (0.75,-0.8665);
    \draw[] (-0.75,0.8665) -- (-0.75,1.732);
    \draw[] (0.75,0.8665) -- (0.75,1.732);
    \draw[] (-0.75,-0.43325) node[left] {$C_1$};
    \draw[] (-0.75,1.29975) node[left] {$C_3$};
    \draw[] (0.75,-0.43325) node[right] {$C_2$};
    \draw[] (0.75,1.29975) node[right] {$C_4$};
  \end{tikzpicture}.
\end{equation}
This represents a composition of all $\mc_U$ that are supported in the circuit qubits in $C_1 \cup C_2 = C_3 \cup C_4$ in the $i$th scale layer. The number of legs on this tensor does not matter, as long as the input (up) matches with the output (down).

Diagrammatically, we claim that
\begin{equation}
  \begin{tikzpicture}[scale=0.55,baseline={([yshift=-.5ex]current bounding box.center)}]
    \draw[thick] (-1,0)--(1,0)--(1,0.8665)--(-1,0.8665)--cycle;
    \draw[] (0.75,0.8665) -- ++ (0,0.8665);
    \draw[] (0,0.43325) node[] {$\mc_1$};
    \draw[thick] (0.75,0.8665+0.8665) -- ++ (0.5,0) -- ++ (-0.5,0.43325) -- ++ (-0.5,-0.43325)--cycle;
    \draw[] (0.75, 0.8665+0.43325) node[right] {$Q_1$};
    \draw[] (-0.75,0) -- ++ (0,-1.732);
    \draw[] (-0.75,-0.43325) node[left] {$Q_1$};

    \begin{scope}[yshift=-1.732cm-0.8665cm]
      \draw[thick] (-1,0)--(1,0)--(1,0.8665)--(-1,0.8665)--cycle;
      \draw[] (0.75,0.8665) -- ++ (0,0.8665);
      \draw[] (0,0.43325) node[] {$\mc_2$};
      \draw[thick] (0.75,0.8665+0.8665) -- ++ (0.5,0) -- ++ (-0.5,0.43325) -- ++ (-0.5,-0.43325)--cycle;
      \draw[] (0.75, 0.8665+0.43325) node[right] {$Q_2$};
      \draw[] (-0.75,0) -- ++ (0,-1.732);
      \draw[] (-0.75,-0.43325) node[left] {$Q_1Q_2$};
    \end{scope}

    \begin{scope}[yshift=-1.732cm-0.8665cm-1.732cm-0.8665cm]
      \draw[] (0.75,0.8665) -- ++ (0,0.8665);
      \draw[thick] (0.75,0.8665+0.8665) -- ++ (0.5,0) -- ++ (-0.5,0.43325) -- ++ (-0.5,-0.43325)--cycle;
      \draw[] (0.75, 0.8665+0.43325) node[right] {$Q_3$};
      \draw[dotted] (0,0)--(0,0.8665);
    \end{scope}

    \begin{scope}[yshift=-5.197cm-0.8665cm]
      \draw[thick] (-1,0)--(1,0)--(1,0.8665)--(-1,0.8665)--cycle;
      \draw[] (0,0.43325) node[] {$\mc_{n-1}$};
      \draw[] (-0.75,0) -- ++ (0,-1.732);
      \draw[] (-0.75,-0.43325) node[left] {$Q_1\ldots Q_{n-1}$};
    \end{scope}

    \begin{scope}[yshift=-5.197cm-1.732cm-1.732cm]
      \draw[thick] (-1,0)--(1,0)--(1,0.8665)--(-1,0.8665)--cycle;
      \draw[] (0.75,0.8665) -- ++ (0,0.8665);
      \draw[] (0,0.43325) node[] {$\mc_n$};
      \draw[thick] (0.75,0.8665+0.8665) -- ++ (0.5,0) -- ++ (-0.5,0.43325) -- ++ (-0.5,-0.43325)--cycle;
      \draw[] (0.75, 0.8665+0.43325) node[right] {$Q_n$};
      \draw[] (-0.75,0) -- ++ (0,-0.8665);
      \draw[] (0.75,0) -- ++ (0,-0.8665);
      \draw[] (-0.75,-0.43325) node[left] {$C(n)$};
      \draw[] (0.75,-0.43325) node[right] {$\overline{C(n)}$};
      \draw[thick] (0.75,-0.8665) -- ++ (0.5,0) -- ++ (-0.5,-0.43325) -- ++ (-0.5,+0.43325)--cycle;
    \end{scope}
  \end{tikzpicture}
  =
  \begin{tikzpicture}[scale=0.55,baseline={([yshift=-.5ex]current bounding box.center)}]
    \draw[thick] (-1,0)--(1,0)--(1,0.8665)--(-1,0.8665)--cycle;
    \draw[] (0.75,0.8665) -- ++ (0,0.8665);
    \draw[] (0,0.43325) node[] {$\mc_1^{\mathfrak{A}}$};
    \draw[thick] (0.75,0.8665+0.8665) -- ++ (0.5,0) -- ++ (-0.5,0.43325) -- ++ (-0.5,-0.43325)--cycle;
    \draw[] (0.75, 0.8665+0.43325) node[right] {$P(1)\overline{P(1)}$};
    \draw[] (-0.75,0) -- ++ (0,-1.732);
    \draw[] (-0.75,-0.8665) node[left] {$P(1)$};
    \draw[] (0.75,0) -- ++ (0,-0.43325);
    \draw[] (0.75,-0.43325) -- ++ (0.5,0) -- ++ (-0.5,-0.43325) -- ++ (-0.5,0.43325) -- cycle;

    \begin{scope}[yshift=-1.732cm-0.8665cm]
      \draw[thick] (-1,0)--(1,0)--(1,0.8665)--(-1,0.8665)--cycle;
      \draw[] (0.75,0.8665) -- ++ (0,0.43325);
      \draw[] (0,0.43325) node[] {$\mc_2^{\mathfrak{A}}$};
      \draw[thick] (0.75,0.8665+0.43325) -- ++ (0.5,0) -- ++ (-0.5,0.43325) -- ++ (-0.5,-0.43325)--cycle;
      \draw[] (1, 1.732) node[right] {$\overline{P(1)}$};
      \draw[] (-0.75,0) -- ++ (0,-1.732);
      \draw[] (-0.75,-0.8665) node[left] {$P(2)$};
      \draw[] (0.75,0) -- ++ (0,-0.43325);
      \draw[] (0.75,-0.43325) -- ++ (0.5,0) -- ++ (-0.5,-0.43325) -- ++ (-0.5,0.43325) -- cycle;

    \end{scope}

    \begin{scope}[yshift=-1.732cm-0.8665cm-1.732cm-0.8665cm]
      \draw[] (0.75,0.8665) -- ++ (0,0.43325);
      \draw[thick] (0.75,0.8665+0.43325) -- ++ (0.5,0) -- ++ (-0.5,0.43325) -- ++ (-0.5,-0.43325)--cycle;
      \draw[] (1, 1.732) node[right] {$\overline{P(2)}$};
      \draw[dotted] (0,0)--(0,0.8665);
    \end{scope}

    \begin{scope}[yshift=-5.197cm-0.8665cm]
      \draw[thick] (-1,0)--(1,0)--(1,0.8665)--(-1,0.8665)--cycle;
      \draw[] (0,0.43325) node[] {$\mc_{n-1}^{\mathfrak{A}}$};
      \draw[] (-0.75,0) -- ++ (0,-1.732);
      \draw[] (-0.75,-0.8665) node[left] {$P(n-1)$};
      \draw[] (0.75,0) -- ++ (0,-0.43325);
      \draw[] (0.75,-0.43325) -- ++ (0.5,0) -- ++ (-0.5,-0.43325) -- ++ (-0.5,0.43325) -- cycle;
    \end{scope}

    \begin{scope}[yshift=-5.197cm-1.732cm-1.732cm]
      \draw[thick] (-1,0)--(1,0)--(1,0.8665)--(-1,0.8665)--cycle;
      \draw[] (0.75,0.8665) -- ++ (0,0.43325);
      \draw[] (0,0.43325) node[] {$\mc_n^{\mathfrak{A}}$};
      \draw[thick] (0.75,0.8665+0.43325) -- ++ (0.5,0) -- ++ (-0.5,0.43325) -- ++ (-0.5,-0.43325)--cycle;
      \draw[] (1, 1.732) node[right] {$\overline{P(n-1)}$};
      \draw[] (-0.75,0) -- ++ (0,-0.8665);
      \draw[] (0.75,0) -- ++ (0,-0.8665);
      \draw[] (-0.75,-0.43325) node[left] {$P(n)$};
      \draw[] (1,-0.43325) node[right] {$\overline{P(n)}$};
      \draw[thick] (0.75,-1.732-0.8665) -- ++ (0.5,0) -- ++ (-0.5,-0.43325) -- ++ (-0.5,+0.43325)--cycle;
      \draw[thick] (-1,-0.8665) -- ++ (2,0) -- ++ (0,-0.8665) -- ++ (-2,0) --cycle;
      \draw[] (-0.75,-1.732) -- ++ (0,-0.8665);
      \draw[] (0.75,-1.732) -- ++ (0,-0.8665);
      \draw[] (0,-1.732+0.43325) node[] {$\mathfrak{A}^{-1}$};
      \draw[] (-0.75,-1.732-0.43325) node[left] {$C(n)$};
      \draw[] (0.75,-1.732-0.43325) node[right] {$C'(n)\setminus C(n)$};
    \end{scope}
  \end{tikzpicture},
  \label{eq:variational_claim}
\end{equation}
where $Q_i$ is a set of qubits that are newly introduced in the $i$th scale layer and $P(i) = \mathfrak{A}(C(i))$. We use $\overline{X}$ to denote the complement of $X$. Specifically, $\overline{C(i)}$ is the set of circuit qubits that are not in $C(i)$ and $\overline{P(i)}$ is a set of physical qubits that are not in $P(i)$. Also,
\begin{equation}
  \begin{tikzpicture}[scale=0.55,baseline={([yshift=-.5ex]current bounding box.center)}]
    \draw[thick] (-1,0)--(1,0)--(1,0.8665)--(-1,0.8665)--cycle;
    \draw[] (0,0.43325) node[] {$\mc_i^{\mathfrak{A}}$};
    \draw[] (0.75,0.8665) -- ++ (0,0.8665);
    \draw[] (-0.75,0) -- ++ (0,-0.8665);
    \draw[] (0.75,0) -- ++ (0,-0.8665);
    \draw[] (-0.75,0.8665) -- ++ (0,0.8665);
    \draw[] (0.75,0.8665+0.43325) node[right] {$\overline{P(i-1)}$};
    \draw[] (0.75,-0.43325) node[right] {$\overline{P(i)}$};
    \draw[] (-0.75,0.8665+0.43325) node[left] {$P(i-1)$};
    \draw[] (-0.75,-0.43325) node[left] {$P(i)$};
  \end{tikzpicture}
  =
  \begin{tikzpicture}[scale=0.55,baseline={([yshift=-.5ex]current bounding box.center)}]
    \draw[thick] (-1,0)--(1,0)--(1,0.8665)--(-1,0.8665)--cycle;
    \draw[thick] (-1,1.732) -- ++ (2,0) -- ++ (0,0.8665) -- ++ (-2,0) -- cycle;
    \draw[thick] (-1,-1.732) -- ++ (2,0) -- ++ (0,0.8665) -- ++ (-2,0) -- cycle;
    \draw[] (0,0.43325) node[] {$\mc_i$};
    \draw[] (0,0.43325+1.732) node[] {$\mathfrak{A}^{-1}$};
    \draw[] (0,0.43325-1.732) node[] {$\mathfrak{A}$};

    \draw[] (0.75,0.8665) -- ++ (0,0.8665);
    \draw[] (-0.75,0) -- ++ (0,-0.8665);
    \draw[] (0.75,0) -- ++ (0,-0.8665);
    \draw[] (-0.75,0.8665) -- ++ (0,0.8665);

    \draw[] (0.75,1.732+0.8665) -- ++ (0,0.8665);
    \draw[] (-0.75,1.732+0.8665) -- ++ (0,0.8665);
    \draw[] (0.75,-1.732) -- ++ (0,-0.8665);
    \draw[] (-0.75,-1.732) -- ++ (0,-0.8665);

    \draw[] (0.75,0.8665+0.43325) node[right] {$D(i-1)$};
    \draw[] (0.75,-0.43325) node[right] {$D(i)$};
    \draw[] (-0.75,0.8665+0.43325) node[left] {$C(i-1)$};
    \draw[] (-0.75,-0.43325) node[left] {$C(i)$};

    \draw[] (0.75, 1.732+0.8665+0.43325) node[right] {$\overline{P(i-1)}$};
    \draw[] (-0.75, 1.732+0.8665+0.43325) node[left] {$P(i-1)$};
    \draw[] (0.75, -1.732-0.8665+0.43325) node[right] {$\overline{P(i)}$};
    \draw[] (-0.75, -1.732-0.8665+0.43325) node[left] {$P(i)$};
  \end{tikzpicture}
  \label{eq:channel_embedding}
\end{equation}
where $D(i)= C'(i) \setminus C(i)$ and
\begin{equation}
  \begin{aligned}
  \bigotimes_{p\in P(i)} \rho_0^{(p)} &=
  \begin{tikzpicture}[scale=0.5,baseline={([yshift=-.5ex]current bounding box.center)}]
    \draw[thick] (-1,0)--(1,0)--(0,0.8665)--cycle;
    \draw[] (0,0) -- (0,-0.8665);
    \draw[] (0,-0.43325) node[left] {$P(i)$};
  \end{tikzpicture},
  &\Tr_{P(i)} &=
  \begin{tikzpicture}[scale=0.5,baseline={([yshift=-.5ex]current bounding box.center)}]
    \draw[thick] (-1,0)--(1,0)--(0,-0.8665)--cycle;
    \draw[] (0,0) -- (0,0.8665);
    \draw[] (0,0.43325) node[left] {$P(i)$};
  \end{tikzpicture}.
\end{aligned}
\end{equation}
Here the box containing $\mathfrak{A}$ is an isometric embedding of each circuit qubit in $C'(i)$ to the assigned physical qubits in $P(i)\overline{P(i)}$ and $\mathfrak{A}^{-1}$ is its inverse. Note that the inverse is well-defined because all the circuit qubits in $C'(i)$ are assigned with different physical qubits.

The left hand side of Eq.~\eqref{eq:variational_claim} represents the reduced density matrix of the noisy DMERA over a subsystem $C(n)$. The right hand side is the reduced density matrix obtained from our algorithm over a subsystem $P(n) = \mathfrak{A}(n)$ which is then isometrically embedded into $C(n)$. Thus, proving Eq.~\eqref{eq:variational_claim} for every $C(n)$ amounts to proving our claim.

To derive Eq.~\eqref{eq:variational_claim}, note that the causal structure of the noisy DMERA circuit is identical to that of its noiseless counterpart. Therefore, the following identity holds:
\begin{equation}
    \begin{tikzpicture}[scale=0.55,baseline={([yshift=-.5ex]current bounding box.center)}]
    \draw[thick] (-1,0)--(1,0)--(1,0.8665)--(-1,0.8665)--cycle;
    \draw[] (0.75,0.8665) -- ++ (0,0.8665);
    \draw[] (0,0.43325) node[] {$\mc_1$};
    \draw[thick] (0.75,0.8665+0.8665) -- ++ (0.5,0) -- ++ (-0.5,0.43325) -- ++ (-0.5,-0.43325)--cycle;
    \draw[] (0.75, 0.8665+0.43325) node[right] {$Q_1$};
    \draw[] (-0.75,0) -- ++ (0,-1.732);
    \draw[] (-0.75,-0.43325) node[left] {$Q_1$};

    \begin{scope}[yshift=-1.732cm-0.8665cm]
      \draw[thick] (-1,0)--(1,0)--(1,0.8665)--(-1,0.8665)--cycle;
      \draw[] (0.75,0.8665) -- ++ (0,0.8665);
      \draw[] (0,0.43325) node[] {$\mc_2$};
      \draw[thick] (0.75,0.8665+0.8665) -- ++ (0.5,0) -- ++ (-0.5,0.43325) -- ++ (-0.5,-0.43325)--cycle;
      \draw[] (0.75, 0.8665+0.43325) node[right] {$Q_2$};
      \draw[] (-0.75,0) -- ++ (0,-1.732);
      \draw[] (-0.75,-0.43325) node[left] {$Q_1Q_2$};
    \end{scope}

    \begin{scope}[yshift=-1.732cm-0.8665cm-1.732cm-0.8665cm]
      \draw[] (0.75,0.8665) -- ++ (0,0.8665);
      \draw[thick] (0.75,0.8665+0.8665) -- ++ (0.5,0) -- ++ (-0.5,0.43325) -- ++ (-0.5,-0.43325)--cycle;
      \draw[] (0.75, 0.8665+0.43325) node[right] {$Q_3$};
      \draw[dotted] (0,0)--(0,0.8665);
    \end{scope}

    \begin{scope}[yshift=-5.197cm-0.8665cm]
      \draw[thick] (-1,0)--(1,0)--(1,0.8665)--(-1,0.8665)--cycle;
      \draw[] (0,0.43325) node[] {$\mc_{n-1}$};
      \draw[] (-0.75,0) -- ++ (0,-1.732);
      \draw[] (-0.75,-0.43325) node[left] {$Q_1\ldots Q_{n-1}$};
    \end{scope}

    \begin{scope}[yshift=-5.197cm-1.732cm-1.732cm]
      \draw[thick] (-1,0)--(1,0)--(1,0.8665)--(-1,0.8665)--cycle;
      \draw[] (0.75,0.8665) -- ++ (0,0.8665);
      \draw[] (0,0.43325) node[] {$\mc_n$};
      \draw[thick] (0.75,0.8665+0.8665) -- ++ (0.5,0) -- ++ (-0.5,0.43325) -- ++ (-0.5,-0.43325)--cycle;
      \draw[] (0.75, 0.8665+0.43325) node[right] {$Q_n$};
      \draw[] (-0.75,0) -- ++ (0,-0.8665);
      \draw[] (0.75,0) -- ++ (0,-0.8665);
      \draw[] (-0.75,-0.43325) node[left] {$C(n)$};
      \draw[] (0.75,-0.43325) node[right] {$\overline{C(n)}$};
      \draw[thick] (0.75,-0.8665) -- ++ (0.5,0) -- ++ (-0.5,-0.43325) -- ++ (-0.5,+0.43325)--cycle;
    \end{scope}
  \end{tikzpicture}
  =
    \begin{tikzpicture}[scale=0.55,baseline={([yshift=-.5ex]current bounding box.center)}]
    \draw[thick] (-1,0)--(1,0)--(1,0.8665)--(-1,0.8665)--cycle;
    \draw[] (0.75,0.8665) -- ++ (0,0.8665);
    \draw[] (0,0.43325) node[] {$\mc_1$};
    \draw[thick] (0.75,0.8665+0.8665) -- ++ (0.5,0) -- ++ (-0.5,0.43325) -- ++ (-0.5,-0.43325)--cycle;
    \draw[] (0.75, 0.8665+0.43325) node[right] {$C(1')$};
    \draw[] (-0.75,0) -- ++ (0,-1.732);
    \draw[] (-0.75,-0.8665) node[left] {$C(1)$};
    \draw[] (0.75,0) -- ++ (0,-0.43325);
    \draw[] (0.75,-0.43325) -- ++ (0.5,0) -- ++ (-0.5,-0.43325) -- ++ (-0.5,0.43325) -- cycle;

    \begin{scope}[yshift=-1.732cm-0.8665cm]
      \draw[thick] (-1,0)--(1,0)--(1,0.8665)--(-1,0.8665)--cycle;
      \draw[] (0.75,0.8665) -- ++ (0,0.43325);
      \draw[] (0,0.43325) node[] {$\mc_2$};
      \draw[thick] (0.75,0.8665+0.43325) -- ++ (0.5,0) -- ++ (-0.5,0.43325) -- ++ (-0.5,-0.43325)--cycle;
      \draw[] (1, 1.732+0.43325) node[right] {$D(1)$};
      \draw[] (1, 1.732-0.43325) node[right] {$E(1)$};
      \draw[] (-0.75,0) -- ++ (0,-1.732);
      \draw[] (-0.75,-0.8665) node[left] {$C(2)$};
      \draw[] (0.75,0) -- ++ (0,-0.43325);
      \draw[] (0.75,-0.43325) -- ++ (0.5,0) -- ++ (-0.5,-0.43325) -- ++ (-0.5,0.43325) -- cycle;

    \end{scope}

    \begin{scope}[yshift=-1.732cm-0.8665cm-1.732cm-0.8665cm]
      \draw[] (0.75,0.8665) -- ++ (0,0.43325);
      \draw[thick] (0.75,0.8665+0.43325) -- ++ (0.5,0) -- ++ (-0.5,0.43325) -- ++ (-0.5,-0.43325)--cycle;
      \draw[] (1, 1.732+0.43325) node[right] {$D(2)$};
      \draw[] (1, 1.732-0.43325) node[right] {$E(2)$};
      \draw[dotted] (0,0)--(0,0.8665);
    \end{scope}

    \begin{scope}[yshift=-5.197cm-0.8665cm]
      \draw[thick] (-1,0)--(1,0)--(1,0.8665)--(-1,0.8665)--cycle;
      \draw[] (0,0.43325) node[] {$\mc_{n-1}$};
      \draw[] (-0.75,0) -- ++ (0,-1.732);
      \draw[] (-0.75,-0.8665) node[left] {$C(n-1)$};
      \draw[] (0.75,0) -- ++ (0,-0.43325);
      \draw[] (0.75,-0.43325) -- ++ (0.5,0) -- ++ (-0.5,-0.43325) -- ++ (-0.5,0.43325) -- cycle;
    \end{scope}

    \begin{scope}[yshift=-5.197cm-1.732cm-1.732cm]
      \draw[thick] (-1,0)--(1,0)--(1,0.8665)--(-1,0.8665)--cycle;
      \draw[] (0.75,0.8665) -- ++ (0,0.43325);
      \draw[] (0,0.43325) node[] {$\mc_n$};
      \draw[thick] (0.75,0.8665+0.43325) -- ++ (0.5,0) -- ++ (-0.5,0.43325) -- ++ (-0.5,-0.43325)--cycle;
      \draw[] (1, 1.732+0.43325) node[right] {$D(n-1)$};
      \draw[] (1, 1.732-0.43325) node[right] {$E(n-1)$};
      \draw[] (-0.75,0) -- ++ (0,-0.8665);
      \draw[] (0.75,0) -- ++ (0,-0.8665);
      \draw[] (-0.75,-0.43325) node[left] {$C(n)$};
      \draw[] (1,-0.43325) node[right] {$D(n)$};
      \draw[thick] (0.75,-0.8665) -- ++ (0.5,0) -- ++ (-0.5,-0.43325) -- ++ (-0.5,+0.43325)--cycle;
    \end{scope}
  \end{tikzpicture},
\end{equation}
where $E(i) = C'(i+1) \setminus C(i)$. Also, recall that $D(i) = C'(i) \setminus C(i)$. Here we reduced the circuit to the past causal cone of $C(n)$ and rearranged the partial traces. Applying Eq.~\eqref{eq:channel_embedding}, Eq.~\eqref{eq:variational_claim} is derived.

Let us make a few remarks. First, even though $D(i)$ and $E(i)$ are generally different from each other, $\mathfrak{A}(D(i))= \mathfrak{A}(E(i))$. Consequently, the composition of taking a partial trace on $D(i)$ and preparing a state on $E(i)$ is interpreted as reseting the physical qubits $\mathfrak{A}(D(i))$ to some fixed state. Second, this proof works only because $\mathfrak{A}^{-1}$ is well-defined. Third, the isometric embedding that $\mathfrak{A}$ represents is not actually physically implemented; rather, it is a formal object that relates the reduced density matrices of $C(n)$ to the reduced density matrices prepared by our algorithm. Fourth, the boxes labeled by $\mc_i$ contains spurious gates that do not affect the reduced density matrix of $C(n)$. Of course, the proof is not affected if we remove these gates. Lastly, the proof works for any $C(n)$ because the method of the proof was agnostic about its location. Therefore the entire set of reduced density matrices -- not just a single density matrix -- obtained from the algorithm is consistent with the state we constructed.

We emphasize that Eq.~\eqref{eq:variational_claim} has important practical ramifications. Because the reduced density matrices are consistent with some quantum state, and because the algorithm for minimizing the energy (see Section \ref{section:proposal}) is based on only energy measurements, an experimentalist does not need to precisely characterize the gates to obtain a good variational upper bound to the energy. The only source of error is the measurement error and our assumption that each gate can be modeled by some CPTP map acting on the same set of qubits.

\section{Applications\label{section:applications}}

It is natural to benchmark our scheme on a simple sovlable one-dimensional spin chain, say the quantum Ising model with transverse field with central charge $c=1/2$, to test how well the DMERA can be optimized and how noise-resilient it is in practice in a real scale invariant system. After this benchamarking, there are interesting open questions about scale invariant one-dimensional lattice models with central charge $c>1$ to which DMERA might be usefully applied, e.g., a lattice regulated version of the D1-D5 system at finite central charge.

Another tantalizing application of DMERA is to two- and three-dimensional models that existing tensor network methods currently struggle to address. One example of such a model is the anti-ferromagnetic Heisenberg model on the kagome lattice. The ground state of this model is still being debated, with different numerical methods giving a variety of contradictory answers. Recent DMRG calculations pointed to a gapped topologically ordered ground state \cite{Yan2011,Jiang2012}, but more recent studies have questioned this result, suggesting instead a gapless Dirac spin liquid \cite{He2017}. Many candidate states of the kagome Heisenberg model plausibly have moderate depth DMERA representations, including the Dirac spin liquid and topologically ordered states, so this model is a very natural target for DMERA implementations.

We expect another important application of our proposal would be to the study of two-dimensional Fermi-Hubbard model. As shown by Verstraete and Cirac, it is possible to map the Fermi-Hubbard model to a locally interacting quantum spin model \cite{Verstraete2005}. Furthermore, by using their mapping, one can show that Fermi-Hubbard model with long-range Coulomb interaction can be written as a spin Hamiltonian with local $6$-body interactions and long-range $2$-body interactions. The expectation values of both of these terms can be sampled efficiently in our approach.

To be more specific, let us write down the Hamiltonian. The Fermi-Hubbard model is
\begin{equation}
  H = -t \sum_{<i,j>, \sigma}(c_{i,\sigma}^{\dagger}c_{j,\sigma} + c_{j,\sigma}^{\dagger} c_{i,\sigma})+ U \sum_i n_{i,\uparrow}n_{i,\downarrow}.
\end{equation}
The idea of Ref.~\cite{Verstraete2005} is to introduce a set of auxiliary fermions and apply the Jordan-Wigner transformation along one axis, say $x$. One ends up introducing two qubits for each fermion modes. Let us denote the Pauli operators for these qubits as $X_{i,\sigma,j}$, where $X$ is one of the Pauli-X, Y, or Z operator, $i$ labels the site, $\sigma$ labels $\uparrow$ and $\downarrow$, and $j \in \{1,2 \}$. The horizontal hopping term between $<i,j>$ becomes
\begin{equation}
  -\sum_{\sigma} t (X_{i,\sigma,1}X_{j,\sigma,1} +Y_{i,\sigma,1}Y_{j,\sigma,1})Z_{i,\sigma,2}.
\end{equation}
The vertical hopping term between $<i,j>$ becomes, for $i=(x,y)$,
\begin{equation}
  \begin{aligned}
    x \in o, y\in o &:   -\sum_{\sigma} t(X_{i,\sigma,1}X_{j,\sigma,1} + Y_{i,\sigma,1}Y_{j,\sigma,1})X_{i,\sigma,2}Y_{j,\sigma,2} \\
    x \in o, y\in e &:   +\sum_{\sigma} t(X_{i,\sigma,1}X_{j,\sigma,1} + Y_{i,\sigma,1}Y_{j,\sigma,1})X_{i,\sigma,2}Y_{j,\sigma,2} \\
    x \in e, y\in o &:   -\sum_{\sigma} t(X_{i,\sigma,1}X_{j,\sigma,1} + Y_{i,\sigma,1}Y_{j,\sigma,1})Y_{i,\sigma,2}X_{j,\sigma,2} \\
    x \in e, y\in e &:   +\sum_{\sigma} t(X_{i,\sigma,1}X_{j,\sigma,1} + Y_{i,\sigma,1}Y_{j,\sigma,1})Y_{i,\sigma,2}X_{j,\sigma,2},
  \end{aligned}
\end{equation}
where $x\in o$ means that $x$ is odd and $x\in e$ means that $x$ is even. The interaction term becomes
\begin{equation}
  U(Z_{i,\uparrow,1}-1)(Z_{i,\downarrow,1}-1).
\end{equation}
Lastly, there are terms to be included in order to enforce constraints.
\begin{equation}
  \begin{aligned}
    x \in o, y\in o &:  \sum_{\sigma} Z_{(x+1,y),\sigma,1}Z_{(x,y+1),\sigma,1}Y_{(x,y),\sigma}^{\square} \\
    x \in o, y\in e &:  \sum_{\sigma} Z_{(x+1,y),\sigma,1}Z_{(x,y+1),\sigma,1}X_{(x,y),\sigma}^{\square} \\
    x \in e, y\in o &:  \sum_{\sigma} Z_{(x,y),\sigma,1}Z_{(x+1,y+1),\sigma,1}Y_{(x,y),\sigma}^{\square} \\
    x \in e, y\in e &:  \sum_{\sigma} Z_{(x,y),\sigma,1}Z_{(x+1,y+1),\sigma,1}X_{(x,y),\sigma}^{\square},
  \end{aligned}
\end{equation}
where
\begin{equation}
  \begin{aligned}
  X_{(x,y),\sigma}^{\square} &= X_{(x,y),\sigma,2}X_{(x+1,y),\sigma,2}X_{(x,y+1),\sigma,2}X_{(x+1,y+1),\sigma,2}, \\
  Y_{(x,y),\sigma}^{\square} &= Y_{(x,y),\sigma,2}Y_{(x+1,y),\sigma,2}Y_{(x,y+1),\sigma,2}Y_{(x+1,y+1),\sigma,2}.
\end{aligned}
\end{equation}
\section{Discussion}

One disadvantage of our method, compared to Ref.~\cite{Kim2017,Kim2017a}, is that the circuit that prepares the ansatz is nonlocal. Indeed, one of the strengths of Ref.~\cite{Kim2017,Kim2017a} was that all the gates could retain their geometric locality. Such geometric locality makes the proposal much more experimentally viable than the proposal we have here. However, it should be noted many ion trap architectures can easily handle nonlocal interactions, either by directly applying nonlocal two-qubit gates or by transporting the ions. Very large systems of trapped ions have recently been developed that might be well suited to DMERA \cite{Bohnet2015,Zhang2017}. Alternatively, recent experiments trapping Rydberg atoms in optical tweezers have reported the ability to move the atoms in real time \cite{Bernien2017}. In such a setup, it would also be possible to perform the dynamical interweaving of spins inherent in the definition of DMERA.

Along these lines, there are at least two different ways in which our protocol might be improved. First, if the circuit can be made local without blowing up too much the depth of the circuit, it will make the implementation of our proposal much more experiment-friendly. Second, it is desirable to design a circuit qubit assignment which uses fewer physical qubits, while still ensuring that the energy is variational. Whether these improvements are possible or not is left for future work.

We proved the resilience of DMERA to noise using the using the positivity of non-identity operator scaling dimensions. Another simple physical picture which helps us understand the resilience of DMERA comes from thinking about DMERA as evolution in an expanding universe. In that context, it was observed that the exponential dilution of injected energy keeps the final energy density low, i.e., not growing with circuit depth, even if every layer of the renormalization group circuit is imperfect \cite{Swingle2014}. Interestingly, it can be shown that the same energy dilution argument also applies to branching renormalization group circuits. This observation suggests a branching version of DMERA might also be stable to noise, but at present it remains unclear how to make this rigorous.

It is also interesting to note that sometimes we know the gates by some other means, as in Refs.~\cite{Swingle2016,Haegeman2017}, but the calculation of physical properties using a classical computer is prohibitive due to the large bond dimension. In this case, we could skip the energy minimization step and simply use the quantum computer to approximately calculate physical observables. It might also be interesting to use DMERA circuits to prepare low energy states, e.g. to study low temperature physics, given that cooling is often difficult.

Finally, although we specifically placed our discussion in the context of near term noisy quantum devices, obviously our method would also perform well on a fault tolerant device. In that context, we expect that logical qubits will continue to be a scarce resource, so our method will still be useful since it enables the simulation of very large systems using only modest logical resources.

\textit{Acknowledgements:} BGS was supported in part by the Simons Foundation as part of the It From Qubit collaboration and the National Science Foundation under Grant No. NSF PHY-1125915. IK was supported in part by the Simons Foundation and by MURI grant 553955.
\bibliography{bib}

\begin{thebibliography}{58}%
\makeatletter
\providecommand \@ifxundefined [1]{%
 \@ifx{#1\undefined}
}%
\providecommand \@ifnum [1]{%
 \ifnum #1\expandafter \@firstoftwo
 \else \expandafter \@secondoftwo
 \fi
}%
\providecommand \@ifx [1]{%
 \ifx #1\expandafter \@firstoftwo
 \else \expandafter \@secondoftwo
 \fi
}%
\providecommand \natexlab [1]{#1}%
\providecommand \enquote  [1]{``#1''}%
\providecommand \bibnamefont  [1]{#1}%
\providecommand \bibfnamefont [1]{#1}%
\providecommand \citenamefont [1]{#1}%
\providecommand \href@noop [0]{\@secondoftwo}%
\providecommand \href [0]{\begingroup \@sanitize@url \@href}%
\providecommand \@href[1]{\@@startlink{#1}\@@href}%
\providecommand \@@href[1]{\endgroup#1\@@endlink}%
\providecommand \@sanitize@url [0]{\catcode `\\12\catcode `\$12\catcode
  `\&12\catcode `\#12\catcode `\^12\catcode `\_12\catcode `\%12\relax}%
\providecommand \@@startlink[1]{}%
\providecommand \@@endlink[0]{}%
\providecommand \url  [0]{\begingroup\@sanitize@url \@url }%
\providecommand \@url [1]{\endgroup\@href {#1}{\urlprefix }}%
\providecommand \urlprefix  [0]{URL }%
\providecommand \Eprint [0]{\href }%
\providecommand \doibase [0]{http://dx.doi.org/}%
\providecommand \selectlanguage [0]{\@gobble}%
\providecommand \bibinfo  [0]{\@secondoftwo}%
\providecommand \bibfield  [0]{\@secondoftwo}%
\providecommand \translation [1]{[#1]}%
\providecommand \BibitemOpen [0]{}%
\providecommand \bibitemStop [0]{}%
\providecommand \bibitemNoStop [0]{.\EOS\space}%
\providecommand \EOS [0]{\spacefactor3000\relax}%
\providecommand \BibitemShut  [1]{\csname bibitem#1\endcsname}%
\let\auto@bib@innerbib\@empty
\bibitem [{\citenamefont {Raussendorf}\ and\ \citenamefont
  {Harrington}(2007)}]{Raussendorf2007}%
  \BibitemOpen
  \bibfield  {author} {\bibinfo {author} {\bibfnamefont {R.}~\bibnamefont
  {Raussendorf}}\ and\ \bibinfo {author} {\bibfnamefont {J.}~\bibnamefont
  {Harrington}},\ }\href {\doibase 10.1103/physrevlett.98.190504} {\bibfield
  {journal} {\bibinfo  {journal} {Phys. Rev. Lett.}\ }\textbf {\bibinfo
  {volume} {98}},\ \bibinfo {pages} {190504} (\bibinfo {year}
  {2007})}\BibitemShut {NoStop}%
\bibitem [{\citenamefont {Barends}\ \emph {et~al.}(2016)\citenamefont
  {Barends}, \citenamefont {Shabani}, \citenamefont {Lamata}, \citenamefont
  {Kelly}, \citenamefont {Mezzacapo}, \citenamefont {Heras}, \citenamefont
  {Babbush}, \citenamefont {Fowler}, \citenamefont {Campbell}, \citenamefont
  {Chen}, \citenamefont {Chen}, \citenamefont {Chiaro}, \citenamefont
  {Dunsworth}, \citenamefont {Jeffrey}, \citenamefont {Lucero}, \citenamefont
  {Megrant}, \citenamefont {Mutus}, \citenamefont {Neeley}, \citenamefont
  {Neill}, \citenamefont {O'Malley}, \citenamefont {Quintana}, \citenamefont
  {Roushan}, \citenamefont {Sank}, \citenamefont {Vainsencher}, \citenamefont
  {Wenner}, \citenamefont {White}, \citenamefont {Solano}, \citenamefont
  {Neven},\ and\ \citenamefont {Martinis}}]{Barends2015}%
  \BibitemOpen
  \bibfield  {author} {\bibinfo {author} {\bibfnamefont {R.}~\bibnamefont
  {Barends}}, \bibinfo {author} {\bibfnamefont {A.}~\bibnamefont {Shabani}},
  \bibinfo {author} {\bibfnamefont {L.}~\bibnamefont {Lamata}}, \bibinfo
  {author} {\bibfnamefont {J.}~\bibnamefont {Kelly}}, \bibinfo {author}
  {\bibfnamefont {A.}~\bibnamefont {Mezzacapo}}, \bibinfo {author}
  {\bibfnamefont {U.~L.}\ \bibnamefont {Heras}}, \bibinfo {author}
  {\bibfnamefont {R.}~\bibnamefont {Babbush}}, \bibinfo {author} {\bibfnamefont
  {A.~G.}\ \bibnamefont {Fowler}}, \bibinfo {author} {\bibfnamefont
  {B.}~\bibnamefont {Campbell}}, \bibinfo {author} {\bibfnamefont
  {Y.}~\bibnamefont {Chen}}, \bibinfo {author} {\bibfnamefont {Z.}~\bibnamefont
  {Chen}}, \bibinfo {author} {\bibfnamefont {B.}~\bibnamefont {Chiaro}},
  \bibinfo {author} {\bibfnamefont {A.}~\bibnamefont {Dunsworth}}, \bibinfo
  {author} {\bibfnamefont {E.}~\bibnamefont {Jeffrey}}, \bibinfo {author}
  {\bibfnamefont {E.}~\bibnamefont {Lucero}}, \bibinfo {author} {\bibfnamefont
  {A.}~\bibnamefont {Megrant}}, \bibinfo {author} {\bibfnamefont {J.~Y.}\
  \bibnamefont {Mutus}}, \bibinfo {author} {\bibfnamefont {M.}~\bibnamefont
  {Neeley}}, \bibinfo {author} {\bibfnamefont {C.}~\bibnamefont {Neill}},
  \bibinfo {author} {\bibfnamefont {P.~J.~J.}\ \bibnamefont {O'Malley}},
  \bibinfo {author} {\bibfnamefont {C.}~\bibnamefont {Quintana}}, \bibinfo
  {author} {\bibfnamefont {P.}~\bibnamefont {Roushan}}, \bibinfo {author}
  {\bibfnamefont {D.}~\bibnamefont {Sank}}, \bibinfo {author} {\bibfnamefont
  {A.}~\bibnamefont {Vainsencher}}, \bibinfo {author} {\bibfnamefont
  {J.}~\bibnamefont {Wenner}}, \bibinfo {author} {\bibfnamefont {T.~C.}\
  \bibnamefont {White}}, \bibinfo {author} {\bibfnamefont {E.}~\bibnamefont
  {Solano}}, \bibinfo {author} {\bibfnamefont {H.}~\bibnamefont {Neven}}, \
  and\ \bibinfo {author} {\bibfnamefont {J.~M.}\ \bibnamefont {Martinis}},\
  }\href {\doibase 10.1038/nature17658} {\bibfield  {journal} {\bibinfo
  {journal} {Nature}\ }\textbf {\bibinfo {volume} {534}},\ \bibinfo {pages}
  {222} (\bibinfo {year} {2016})},\ \Eprint {http://arxiv.org/abs/1511.03316v1}
  {1511.03316v1} \BibitemShut {NoStop}%
\bibitem [{\citenamefont {Sheldon}\ \emph {et~al.}(2016)\citenamefont
  {Sheldon}, \citenamefont {Magesan}, \citenamefont {Chow},\ and\ \citenamefont
  {Gambetta}}]{Sheldon2016}%
  \BibitemOpen
  \bibfield  {author} {\bibinfo {author} {\bibfnamefont {S.}~\bibnamefont
  {Sheldon}}, \bibinfo {author} {\bibfnamefont {E.}~\bibnamefont {Magesan}},
  \bibinfo {author} {\bibfnamefont {J.~M.}\ \bibnamefont {Chow}}, \ and\
  \bibinfo {author} {\bibfnamefont {J.~M.}\ \bibnamefont {Gambetta}},\ }\href
  {\doibase 10.1103/physreva.93.060302} {\bibfield  {journal} {\bibinfo
  {journal} {Phys. Rev. A}\ }\textbf {\bibinfo {volume} {93}},\ \bibinfo
  {pages} {060302} (\bibinfo {year} {2016})}\BibitemShut {NoStop}%
\bibitem [{\citenamefont {Ballance}\ \emph {et~al.}(2016)\citenamefont
  {Ballance}, \citenamefont {Harty}, \citenamefont {Linke}, \citenamefont
  {Sepiol},\ and\ \citenamefont {Lucas}}]{Ballance2016}%
  \BibitemOpen
  \bibfield  {author} {\bibinfo {author} {\bibfnamefont {C.~J.}\ \bibnamefont
  {Ballance}}, \bibinfo {author} {\bibfnamefont {T.~P.}\ \bibnamefont {Harty}},
  \bibinfo {author} {\bibfnamefont {N.~M.}\ \bibnamefont {Linke}}, \bibinfo
  {author} {\bibfnamefont {M.~A.}\ \bibnamefont {Sepiol}}, \ and\ \bibinfo
  {author} {\bibfnamefont {D.~M.}\ \bibnamefont {Lucas}},\ }\href {\doibase
  10.1103/physrevlett.117.060504} {\bibfield  {journal} {\bibinfo  {journal}
  {Phys. Rev. Let..}\ }\textbf {\bibinfo {volume} {117}},\ \bibinfo {pages}
  {060504} (\bibinfo {year} {2016})}\BibitemShut {NoStop}%
\bibitem [{\citenamefont {Mourik}\ \emph {et~al.}(2012)\citenamefont {Mourik},
  \citenamefont {Zuo}, \citenamefont {Frolov}, \citenamefont {Plissard},
  \citenamefont {Bakkers},\ and\ \citenamefont {Kouwenhoven}}]{Mourik2012}%
  \BibitemOpen
  \bibfield  {author} {\bibinfo {author} {\bibfnamefont {V.}~\bibnamefont
  {Mourik}}, \bibinfo {author} {\bibfnamefont {K.}~\bibnamefont {Zuo}},
  \bibinfo {author} {\bibfnamefont {S.~M.}\ \bibnamefont {Frolov}}, \bibinfo
  {author} {\bibfnamefont {S.~R.}\ \bibnamefont {Plissard}}, \bibinfo {author}
  {\bibfnamefont {E.~P. A.~M.}\ \bibnamefont {Bakkers}}, \ and\ \bibinfo
  {author} {\bibfnamefont {L.~P.}\ \bibnamefont {Kouwenhoven}},\ }\href
  {\doibase 10.1126/science.1222360} {\bibfield  {journal} {\bibinfo  {journal}
  {Science}\ }\textbf {\bibinfo {volume} {336}},\ \bibinfo {pages} {1003}
  (\bibinfo {year} {2012})}\BibitemShut {NoStop}%
\bibitem [{\citenamefont {Fowler}\ \emph {et~al.}(2012)\citenamefont {Fowler},
  \citenamefont {Mariantoni}, \citenamefont {Martinis},\ and\ \citenamefont
  {Cleland}}]{Fowler2012}%
  \BibitemOpen
  \bibfield  {author} {\bibinfo {author} {\bibfnamefont {A.~G.}\ \bibnamefont
  {Fowler}}, \bibinfo {author} {\bibfnamefont {M.}~\bibnamefont {Mariantoni}},
  \bibinfo {author} {\bibfnamefont {J.~M.}\ \bibnamefont {Martinis}}, \ and\
  \bibinfo {author} {\bibfnamefont {A.~N.}\ \bibnamefont {Cleland}},\ }\href
  {\doibase 10.1103/PhysRevA.86.032324} {\bibfield  {journal} {\bibinfo
  {journal} {Phys. Rev. A}\ }\textbf {\bibinfo {volume} {86}},\ \bibinfo
  {pages} {032324} (\bibinfo {year} {2012})},\ \Eprint
  {http://arxiv.org/abs/1208.0928v2} {1208.0928v2} \BibitemShut {NoStop}%
\bibitem [{\citenamefont {Peruzzo}\ \emph {et~al.}(2014)\citenamefont
  {Peruzzo}, \citenamefont {McClean}, \citenamefont {Shadbolt}, \citenamefont
  {Yung}, \citenamefont {Zhou}, \citenamefont {Love}, \citenamefont
  {Aspuru-Guzik},\ and\ \citenamefont {O'Brien}}]{Peruzzo2013}%
  \BibitemOpen
  \bibfield  {author} {\bibinfo {author} {\bibfnamefont {A.}~\bibnamefont
  {Peruzzo}}, \bibinfo {author} {\bibfnamefont {J.}~\bibnamefont {McClean}},
  \bibinfo {author} {\bibfnamefont {P.}~\bibnamefont {Shadbolt}}, \bibinfo
  {author} {\bibfnamefont {M.-H.}\ \bibnamefont {Yung}}, \bibinfo {author}
  {\bibfnamefont {X.-Q.}\ \bibnamefont {Zhou}}, \bibinfo {author}
  {\bibfnamefont {P.~J.}\ \bibnamefont {Love}}, \bibinfo {author}
  {\bibfnamefont {A.}~\bibnamefont {Aspuru-Guzik}}, \ and\ \bibinfo {author}
  {\bibfnamefont {J.~L.}\ \bibnamefont {O'Brien}},\ }\href {\doibase
  10.1038/ncomms5213} {\bibfield  {journal} {\bibinfo  {journal} {Nature
  Communications}\ }\textbf {\bibinfo {volume} {5}} (\bibinfo {year} {2014}),\
  10.1038/ncomms5213},\ \Eprint {http://arxiv.org/abs/1304.3061v1}
  {1304.3061v1} \BibitemShut {NoStop}%
\bibitem [{\citenamefont {Wecker}\ \emph {et~al.}(2015)\citenamefont {Wecker},
  \citenamefont {Hastings},\ and\ \citenamefont {Troyer}}]{Wecker2015}%
  \BibitemOpen
  \bibfield  {author} {\bibinfo {author} {\bibfnamefont {D.}~\bibnamefont
  {Wecker}}, \bibinfo {author} {\bibfnamefont {M.~B.}\ \bibnamefont
  {Hastings}}, \ and\ \bibinfo {author} {\bibfnamefont {M.}~\bibnamefont
  {Troyer}},\ }\href {\doibase 10.1103/PhysRevA.92.042303} {\bibfield
  {journal} {\bibinfo  {journal} {Phys. Rev. A}\ }\textbf {\bibinfo {volume}
  {92}},\ \bibinfo {pages} {042303} (\bibinfo {year} {2015})},\ \Eprint
  {http://arxiv.org/abs/1507.08969} {arXiv:1507.08969} \BibitemShut {NoStop}%
\bibitem [{\citenamefont {Vidal}(2008)}]{Vidal2008}%
  \BibitemOpen
  \bibfield  {author} {\bibinfo {author} {\bibfnamefont {G.}~\bibnamefont
  {Vidal}},\ }\href {\doibase 10.1103/PhysRevLett.101.110501} {\bibfield
  {journal} {\bibinfo  {journal} {Phys. Rev. Lett.}\ }\textbf {\bibinfo
  {volume} {101}},\ \bibinfo {pages} {110501} (\bibinfo {year}
  {2008})}\BibitemShut {NoStop}%
\bibitem [{\citenamefont {Pednault}\ \emph {et~al.}(2017)\citenamefont
  {Pednault}, \citenamefont {Gunnels}, \citenamefont {Nannicini}, \citenamefont
  {Horesh}, \citenamefont {Magerlein}, \citenamefont {Solomonik},\ and\
  \citenamefont {Wisnieff}}]{Pednault2017}%
  \BibitemOpen
  \bibfield  {author} {\bibinfo {author} {\bibfnamefont {E.}~\bibnamefont
  {Pednault}}, \bibinfo {author} {\bibfnamefont {J.~A.}\ \bibnamefont
  {Gunnels}}, \bibinfo {author} {\bibfnamefont {G.}~\bibnamefont {Nannicini}},
  \bibinfo {author} {\bibfnamefont {L.}~\bibnamefont {Horesh}}, \bibinfo
  {author} {\bibfnamefont {T.}~\bibnamefont {Magerlein}}, \bibinfo {author}
  {\bibfnamefont {E.}~\bibnamefont {Solomonik}}, \ and\ \bibinfo {author}
  {\bibfnamefont {R.}~\bibnamefont {Wisnieff}},\ }\href@noop {} {\  (\bibinfo
  {year} {2017})}\BibitemShut {NoStop}%
\bibitem [{\citenamefont {Kim}(2017{\natexlab{a}})}]{Kim2017}%
  \BibitemOpen
  \bibfield  {author} {\bibinfo {author} {\bibfnamefont {I.~H.}\ \bibnamefont
  {Kim}},\ }\href {http://arxiv.org/abs/1703.00032v1} {\bibfield  {journal}
  {\bibinfo  {journal} {arXiv:1703.00032}\ } (\bibinfo {year}
  {2017}{\natexlab{a}})},\ \Eprint {http://arxiv.org/abs/1703.00032}
  {arXiv:1703.00032 [quant-ph]} \BibitemShut {NoStop}%
\bibitem [{\citenamefont {Kim}(2017{\natexlab{b}})}]{Kim2017a}%
  \BibitemOpen
  \bibfield  {author} {\bibinfo {author} {\bibfnamefont {I.~H.}\ \bibnamefont
  {Kim}},\ }\href {http://arxiv.org/abs/1702.02093v2} {\bibfield  {journal}
  {\bibinfo  {journal} {arXiv:1702.02093}\ } (\bibinfo {year}
  {2017}{\natexlab{b}})},\ \Eprint {http://arxiv.org/abs/1702.02093}
  {arXiv:1702.02093 [quant-ph]} \BibitemShut {NoStop}%
\bibitem [{\citenamefont {Verstraete}\ \emph {et~al.}(2004)\citenamefont
  {Verstraete}, \citenamefont {Garc{\'{i}}a-Ripoll},\ and\ \citenamefont
  {Cirac}}]{Verstraete2004}%
  \BibitemOpen
  \bibfield  {author} {\bibinfo {author} {\bibfnamefont {F.}~\bibnamefont
  {Verstraete}}, \bibinfo {author} {\bibfnamefont {J.~J.}\ \bibnamefont
  {Garc{\'{i}}a-Ripoll}}, \ and\ \bibinfo {author} {\bibfnamefont {J.~I.}\
  \bibnamefont {Cirac}},\ }\href {\doibase 10.1103/PhysRevLett.93.207204}
  {\bibfield  {journal} {\bibinfo  {journal} {Phys. Rev. Lett.}\ }\textbf
  {\bibinfo {volume} {93}},\ \bibinfo {pages} {207204} (\bibinfo {year}
  {2004})}\BibitemShut {NoStop}%
\bibitem [{\citenamefont {{Swingle}}\ and\ \citenamefont
  {{McGreevy}}(2014)}]{Swingle2014}%
  \BibitemOpen
  \bibfield  {author} {\bibinfo {author} {\bibfnamefont {B.}~\bibnamefont
  {{Swingle}}}\ and\ \bibinfo {author} {\bibfnamefont {J.}~\bibnamefont
  {{McGreevy}}},\ }\href@noop {} {\bibfield  {journal} {\bibinfo  {journal}
  {ArXiv e-prints}\ } (\bibinfo {year} {2014})},\ \Eprint
  {http://arxiv.org/abs/1407.8203} {arXiv:1407.8203 [cond-mat.str-el]}
  \BibitemShut {NoStop}%
\bibitem [{\citenamefont {{Evenbly}}\ and\ \citenamefont
  {{Vidal}}(2015)}]{Evenbly2015}%
  \BibitemOpen
  \bibfield  {author} {\bibinfo {author} {\bibfnamefont {G.}~\bibnamefont
  {{Evenbly}}}\ and\ \bibinfo {author} {\bibfnamefont {G.}~\bibnamefont
  {{Vidal}}},\ }\href {\doibase 10.1103/PhysRevLett.115.200401} {\bibfield
  {journal} {\bibinfo  {journal} {Physical Review Letters}\ }\textbf {\bibinfo
  {volume} {115}},\ \bibinfo {eid} {200401} (\bibinfo {year} {2015})},\ \Eprint
  {http://arxiv.org/abs/1502.05385} {arXiv:1502.05385 [cond-mat.str-el]}
  \BibitemShut {NoStop}%
\bibitem [{\citenamefont {{Czech}}\ \emph {et~al.}(2016)\citenamefont
  {{Czech}}, \citenamefont {{Evenbly}}, \citenamefont {{Lamprou}},
  \citenamefont {{McCandlish}}, \citenamefont {{Qi}}, \citenamefont {{Sully}},\
  and\ \citenamefont {{Vidal}}}]{Czech2016}%
  \BibitemOpen
  \bibfield  {author} {\bibinfo {author} {\bibfnamefont {B.}~\bibnamefont
  {{Czech}}}, \bibinfo {author} {\bibfnamefont {G.}~\bibnamefont {{Evenbly}}},
  \bibinfo {author} {\bibfnamefont {L.}~\bibnamefont {{Lamprou}}}, \bibinfo
  {author} {\bibfnamefont {S.}~\bibnamefont {{McCandlish}}}, \bibinfo {author}
  {\bibfnamefont {X.-l.}\ \bibnamefont {{Qi}}}, \bibinfo {author}
  {\bibfnamefont {J.}~\bibnamefont {{Sully}}}, \ and\ \bibinfo {author}
  {\bibfnamefont {G.}~\bibnamefont {{Vidal}}},\ }\href {\doibase
  10.1103/PhysRevB.94.085101} {\bibfield  {journal} {\bibinfo  {journal}
  {\prb}\ }\textbf {\bibinfo {volume} {94}},\ \bibinfo {eid} {085101} (\bibinfo
  {year} {2016})},\ \Eprint {http://arxiv.org/abs/1510.07637} {arXiv:1510.07637
  [cond-mat.str-el]} \BibitemShut {NoStop}%
\bibitem [{\citenamefont {{Swingle}}\ \emph {et~al.}(2016)\citenamefont
  {{Swingle}}, \citenamefont {{McGreevy}},\ and\ \citenamefont
  {{Xu}}}]{Swingle2016}%
  \BibitemOpen
  \bibfield  {author} {\bibinfo {author} {\bibfnamefont {B.}~\bibnamefont
  {{Swingle}}}, \bibinfo {author} {\bibfnamefont {J.}~\bibnamefont
  {{McGreevy}}}, \ and\ \bibinfo {author} {\bibfnamefont {S.}~\bibnamefont
  {{Xu}}},\ }\href {\doibase 10.1103/PhysRevB.93.205159} {\bibfield  {journal}
  {\bibinfo  {journal} {\prb}\ }\textbf {\bibinfo {volume} {93}},\ \bibinfo
  {eid} {205159} (\bibinfo {year} {2016})},\ \Eprint
  {http://arxiv.org/abs/1602.02805} {arXiv:1602.02805 [cond-mat.str-el]}
  \BibitemShut {NoStop}%
\bibitem [{\citenamefont {{Haegeman}}\ \emph {et~al.}(2017)\citenamefont
  {{Haegeman}}, \citenamefont {{Swingle}}, \citenamefont {{Walter}},
  \citenamefont {{Cotler}}, \citenamefont {{Evenbly}},\ and\ \citenamefont
  {{Scholz}}}]{Haegeman2017}%
  \BibitemOpen
  \bibfield  {author} {\bibinfo {author} {\bibfnamefont {J.}~\bibnamefont
  {{Haegeman}}}, \bibinfo {author} {\bibfnamefont {B.}~\bibnamefont
  {{Swingle}}}, \bibinfo {author} {\bibfnamefont {M.}~\bibnamefont {{Walter}}},
  \bibinfo {author} {\bibfnamefont {J.}~\bibnamefont {{Cotler}}}, \bibinfo
  {author} {\bibfnamefont {G.}~\bibnamefont {{Evenbly}}}, \ and\ \bibinfo
  {author} {\bibfnamefont {V.~B.}\ \bibnamefont {{Scholz}}},\ }\href@noop {}
  {\bibfield  {journal} {\bibinfo  {journal} {ArXiv e-prints}\ } (\bibinfo
  {year} {2017})},\ \Eprint {http://arxiv.org/abs/1707.06243} {arXiv:1707.06243
  [quant-ph]} \BibitemShut {NoStop}%
\bibitem [{\citenamefont {Shor}(1997)}]{Shor1997}%
  \BibitemOpen
  \bibfield  {author} {\bibinfo {author} {\bibfnamefont {P.~W.}\ \bibnamefont
  {Shor}},\ }\href@noop {} {\bibfield  {journal} {\bibinfo  {journal} {SIAM J.
  Comput.,}\ }\textbf {\bibinfo {volume} {26}},\ \bibinfo {pages} {1484}
  (\bibinfo {year} {1997})},\ \Eprint {http://arxiv.org/abs/quant-ph/9508027v2}
  {quant-ph/9508027v2} \BibitemShut {NoStop}%
\bibitem [{\citenamefont {Giovannetti}\ \emph {et~al.}(2008)\citenamefont
  {Giovannetti}, \citenamefont {Montangero},\ and\ \citenamefont
  {Fazio}}]{Giovannetti2008}%
  \BibitemOpen
  \bibfield  {author} {\bibinfo {author} {\bibfnamefont {V.}~\bibnamefont
  {Giovannetti}}, \bibinfo {author} {\bibfnamefont {S.}~\bibnamefont
  {Montangero}}, \ and\ \bibinfo {author} {\bibfnamefont {R.}~\bibnamefont
  {Fazio}},\ }\href {\doibase 10.1103/physrevlett.101.180503} {\bibfield
  {journal} {\bibinfo  {journal} {Phys. Rev. Lett.}\ }\textbf {\bibinfo
  {volume} {101}},\ \bibinfo {pages} {180503} (\bibinfo {year}
  {2008})}\BibitemShut {NoStop}%
\bibitem [{\citenamefont {Evenbly}\ and\ \citenamefont
  {Vidal}(2009)}]{Evenbly2009}%
  \BibitemOpen
  \bibfield  {author} {\bibinfo {author} {\bibfnamefont {G.}~\bibnamefont
  {Evenbly}}\ and\ \bibinfo {author} {\bibfnamefont {G.}~\bibnamefont
  {Vidal}},\ }\href {\doibase 10.1103/physrevb.79.144108} {\bibfield  {journal}
  {\bibinfo  {journal} {Phys. Rev. B}\ }\textbf {\bibinfo {volume} {79}}
  (\bibinfo {year} {2009}),\ 10.1103/physrevb.79.144108}\BibitemShut {NoStop}%
\bibitem [{\citenamefont {Kim}\ and\ \citenamefont
  {Kastoryano}(2017)}]{Kim2017b}%
  \BibitemOpen
  \bibfield  {author} {\bibinfo {author} {\bibfnamefont {I.~H.}\ \bibnamefont
  {Kim}}\ and\ \bibinfo {author} {\bibfnamefont {M.~J.}\ \bibnamefont
  {Kastoryano}},\ }\href {\doibase 10.1007/jhep04(2017)040} {\bibfield
  {journal} {\bibinfo  {journal} {JHEP}\ }\textbf {\bibinfo {volume} {2017}},\
  \bibinfo {pages} {40} (\bibinfo {year} {2017})}\BibitemShut {NoStop}%
\bibitem [{\citenamefont {Spall}(1992)}]{Spall1992}%
  \BibitemOpen
  \bibfield  {author} {\bibinfo {author} {\bibfnamefont {J.}~\bibnamefont
  {Spall}},\ }\href {\doibase 10.1109/9.119632} {\bibfield  {journal} {\bibinfo
   {journal} {IEEE Transactions on Automatic Control}\ }\textbf {\bibinfo
  {volume} {37}},\ \bibinfo {pages} {332} (\bibinfo {year} {1992})}\BibitemShut
  {NoStop}%
\bibitem [{Note1()}]{Note1}%
  \BibitemOpen
  \bibinfo {note} {An example would be a $v_i=\pm 1$ with probability $\protect
  \frac {1}{2}$ for both outcomes. Note that the normal distribution does not
  satisfy this condition.}\BibitemShut {Stop}%
\bibitem [{\citenamefont {{Granade}}\ \emph {et~al.}(2016)\citenamefont
  {{Granade}}, \citenamefont {{Ferrie}},\ and\ \citenamefont
  {{Flammia}}}]{Granade2016}%
  \BibitemOpen
  \bibfield  {author} {\bibinfo {author} {\bibfnamefont {C.}~\bibnamefont
  {{Granade}}}, \bibinfo {author} {\bibfnamefont {C.}~\bibnamefont {{Ferrie}}},
  \ and\ \bibinfo {author} {\bibfnamefont {S.~T.}\ \bibnamefont {{Flammia}}},\
  }\href@noop {} {\bibfield  {journal} {\bibinfo  {journal} {ArXiv e-prints}\ }
  (\bibinfo {year} {2016})},\ \Eprint {http://arxiv.org/abs/1605.05039}
  {arXiv:1605.05039 [quant-ph]} \BibitemShut {NoStop}%
\bibitem [{\citenamefont {Rizzi}\ \emph {et~al.}(2010)\citenamefont {Rizzi},
  \citenamefont {Montangero}, \citenamefont {Silvi}, \citenamefont
  {Giovannetti},\ and\ \citenamefont {Fazio}}]{Rizzi2010}%
  \BibitemOpen
  \bibfield  {author} {\bibinfo {author} {\bibfnamefont {M.}~\bibnamefont
  {Rizzi}}, \bibinfo {author} {\bibfnamefont {S.}~\bibnamefont {Montangero}},
  \bibinfo {author} {\bibfnamefont {P.}~\bibnamefont {Silvi}}, \bibinfo
  {author} {\bibfnamefont {V.}~\bibnamefont {Giovannetti}}, \ and\ \bibinfo
  {author} {\bibfnamefont {R.}~\bibnamefont {Fazio}},\ }\href
  {http://stacks.iop.org/1367-2630/12/i=7/a=075018} {\bibfield  {journal}
  {\bibinfo  {journal} {New Journal of Physics}\ }\textbf {\bibinfo {volume}
  {12}},\ \bibinfo {pages} {075018} (\bibinfo {year} {2010})}\BibitemShut
  {NoStop}%
\bibitem [{\citenamefont {{Evenbly}}\ and\ \citenamefont
  {{Vidal}}(2011{\natexlab{a}})}]{Evenbly2011b}%
  \BibitemOpen
  \bibfield  {author} {\bibinfo {author} {\bibfnamefont {G.}~\bibnamefont
  {{Evenbly}}}\ and\ \bibinfo {author} {\bibfnamefont {G.}~\bibnamefont
  {{Vidal}}},\ }\href@noop {} {\bibfield  {journal} {\bibinfo  {journal} {ArXiv
  e-prints}\ } (\bibinfo {year} {2011}{\natexlab{a}})},\ \Eprint
  {http://arxiv.org/abs/1109.5334} {arXiv:1109.5334 [quant-ph]} \BibitemShut
  {NoStop}%
\bibitem [{\citenamefont {{Corboz}}\ \emph {et~al.}(2010)\citenamefont
  {{Corboz}}, \citenamefont {{Evenbly}}, \citenamefont {{Verstraete}},\ and\
  \citenamefont {{Vidal}}}]{Corboz2010}%
  \BibitemOpen
  \bibfield  {author} {\bibinfo {author} {\bibfnamefont {P.}~\bibnamefont
  {{Corboz}}}, \bibinfo {author} {\bibfnamefont {G.}~\bibnamefont {{Evenbly}}},
  \bibinfo {author} {\bibfnamefont {F.}~\bibnamefont {{Verstraete}}}, \ and\
  \bibinfo {author} {\bibfnamefont {G.}~\bibnamefont {{Vidal}}},\ }\href
  {\doibase 10.1103/PhysRevA.81.010303} {\bibfield  {journal} {\bibinfo
  {journal} {\pra}\ }\textbf {\bibinfo {volume} {81}},\ \bibinfo {eid} {010303}
  (\bibinfo {year} {2010})},\ \Eprint {http://arxiv.org/abs/0904.4151}
  {arXiv:0904.4151 [cond-mat.str-el]} \BibitemShut {NoStop}%
\bibitem [{\citenamefont {{Pineda}}\ \emph {et~al.}(2010)\citenamefont
  {{Pineda}}, \citenamefont {{Barthel}},\ and\ \citenamefont
  {{Eisert}}}]{Pineda2010}%
  \BibitemOpen
  \bibfield  {author} {\bibinfo {author} {\bibfnamefont {C.}~\bibnamefont
  {{Pineda}}}, \bibinfo {author} {\bibfnamefont {T.}~\bibnamefont {{Barthel}}},
  \ and\ \bibinfo {author} {\bibfnamefont {J.}~\bibnamefont {{Eisert}}},\
  }\href {\doibase 10.1103/PhysRevA.81.050303} {\bibfield  {journal} {\bibinfo
  {journal} {\pra}\ }\textbf {\bibinfo {volume} {81}},\ \bibinfo {eid} {050303}
  (\bibinfo {year} {2010})},\ \Eprint {http://arxiv.org/abs/0905.0669}
  {arXiv:0905.0669 [quant-ph]} \BibitemShut {NoStop}%
\bibitem [{\citenamefont {{Maldacena}}(1999)}]{Maldacena1998}%
  \BibitemOpen
  \bibfield  {author} {\bibinfo {author} {\bibfnamefont {J.}~\bibnamefont
  {{Maldacena}}},\ }\href {\doibase 10.1023/A:1026654312961} {\bibfield
  {journal} {\bibinfo  {journal} {International Journal of Theoretical
  Physics}\ }\textbf {\bibinfo {volume} {38}},\ \bibinfo {pages} {1113}
  (\bibinfo {year} {1999})},\ \Eprint {http://arxiv.org/abs/hep-th/9711200}
  {hep-th/9711200} \BibitemShut {NoStop}%
\bibitem [{\citenamefont {{Gubser}}\ \emph {et~al.}(1998)\citenamefont
  {{Gubser}}, \citenamefont {{Klebanov}},\ and\ \citenamefont
  {{Polyakov}}}]{Gubser1998}%
  \BibitemOpen
  \bibfield  {author} {\bibinfo {author} {\bibfnamefont {S.~S.}\ \bibnamefont
  {{Gubser}}}, \bibinfo {author} {\bibfnamefont {I.~R.}\ \bibnamefont
  {{Klebanov}}}, \ and\ \bibinfo {author} {\bibfnamefont {A.~M.}\ \bibnamefont
  {{Polyakov}}},\ }\href {\doibase 10.1016/S0370-2693(98)00377-3} {\bibfield
  {journal} {\bibinfo  {journal} {Physics Letters B}\ }\textbf {\bibinfo
  {volume} {428}},\ \bibinfo {pages} {105} (\bibinfo {year} {1998})},\ \Eprint
  {http://arxiv.org/abs/hep-th/9802109} {hep-th/9802109} \BibitemShut {NoStop}%
\bibitem [{\citenamefont {{Witten}}(1998)}]{Witten1998}%
  \BibitemOpen
  \bibfield  {author} {\bibinfo {author} {\bibfnamefont {E.}~\bibnamefont
  {{Witten}}},\ }\href@noop {} {\bibfield  {journal} {\bibinfo  {journal}
  {Advances in Theoretical and Mathematical Physics}\ }\textbf {\bibinfo
  {volume} {2}},\ \bibinfo {pages} {253} (\bibinfo {year} {1998})},\ \Eprint
  {http://arxiv.org/abs/hep-th/9802150} {hep-th/9802150} \BibitemShut {NoStop}%
\bibitem [{\citenamefont {{Brown}}\ \emph {et~al.}(2015)\citenamefont
  {{Brown}}, \citenamefont {{Roberts}}, \citenamefont {{Susskind}},
  \citenamefont {{Swingle}},\ and\ \citenamefont {{Zhao}}}]{Brown2015a}%
  \BibitemOpen
  \bibfield  {author} {\bibinfo {author} {\bibfnamefont {A.~R.}\ \bibnamefont
  {{Brown}}}, \bibinfo {author} {\bibfnamefont {D.~A.}\ \bibnamefont
  {{Roberts}}}, \bibinfo {author} {\bibfnamefont {L.}~\bibnamefont
  {{Susskind}}}, \bibinfo {author} {\bibfnamefont {B.}~\bibnamefont
  {{Swingle}}}, \ and\ \bibinfo {author} {\bibfnamefont {Y.}~\bibnamefont
  {{Zhao}}},\ }\href@noop {} {\bibfield  {journal} {\bibinfo  {journal} {ArXiv
  e-prints}\ } (\bibinfo {year} {2015})},\ \Eprint
  {http://arxiv.org/abs/1509.07876} {arXiv:1509.07876 [hep-th]} \BibitemShut
  {NoStop}%
\bibitem [{\citenamefont {{Brown}}\ \emph {et~al.}(2016)\citenamefont
  {{Brown}}, \citenamefont {{Roberts}}, \citenamefont {{Susskind}},
  \citenamefont {{Swingle}},\ and\ \citenamefont {{Zhao}}}]{Brown2015b}%
  \BibitemOpen
  \bibfield  {author} {\bibinfo {author} {\bibfnamefont {A.~R.}\ \bibnamefont
  {{Brown}}}, \bibinfo {author} {\bibfnamefont {D.~A.}\ \bibnamefont
  {{Roberts}}}, \bibinfo {author} {\bibfnamefont {L.}~\bibnamefont
  {{Susskind}}}, \bibinfo {author} {\bibfnamefont {B.}~\bibnamefont
  {{Swingle}}}, \ and\ \bibinfo {author} {\bibfnamefont {Y.}~\bibnamefont
  {{Zhao}}},\ }\href {\doibase 10.1103/PhysRevD.93.086006} {\bibfield
  {journal} {\bibinfo  {journal} {\prd}\ }\textbf {\bibinfo {volume} {93}},\
  \bibinfo {eid} {086006} (\bibinfo {year} {2016})},\ \Eprint
  {http://arxiv.org/abs/1512.04993} {arXiv:1512.04993 [hep-th]} \BibitemShut
  {NoStop}%
\bibitem [{\citenamefont {Swingle}(2012)}]{Swingle2009}%
  \BibitemOpen
  \bibfield  {author} {\bibinfo {author} {\bibfnamefont {B.}~\bibnamefont
  {Swingle}},\ }\href {\doibase 10.1103/PhysRevD.86.065007} {\bibfield
  {journal} {\bibinfo  {journal} {Phys. Rev. D}\ }\textbf {\bibinfo {volume}
  {86}},\ \bibinfo {pages} {065007} (\bibinfo {year} {2012})}\BibitemShut
  {NoStop}%
\bibitem [{\citenamefont {{Evenbly}}\ and\ \citenamefont
  {{Vidal}}(2011{\natexlab{b}})}]{Evenbly2011}%
  \BibitemOpen
  \bibfield  {author} {\bibinfo {author} {\bibfnamefont {G.}~\bibnamefont
  {{Evenbly}}}\ and\ \bibinfo {author} {\bibfnamefont {G.}~\bibnamefont
  {{Vidal}}},\ }\href {\doibase 10.1007/s10955-011-0237-4} {\bibfield
  {journal} {\bibinfo  {journal} {Journal of Statistical Physics}\ }\textbf
  {\bibinfo {volume} {145}},\ \bibinfo {pages} {891} (\bibinfo {year}
  {2011}{\natexlab{b}})},\ \Eprint {http://arxiv.org/abs/1106.1082}
  {arXiv:1106.1082 [quant-ph]} \BibitemShut {NoStop}%
\bibitem [{\citenamefont {{Nozaki}}\ \emph {et~al.}(2012)\citenamefont
  {{Nozaki}}, \citenamefont {{Ryu}},\ and\ \citenamefont
  {{Takayanagi}}}]{Nozaki2012}%
  \BibitemOpen
  \bibfield  {author} {\bibinfo {author} {\bibfnamefont {M.}~\bibnamefont
  {{Nozaki}}}, \bibinfo {author} {\bibfnamefont {S.}~\bibnamefont {{Ryu}}}, \
  and\ \bibinfo {author} {\bibfnamefont {T.}~\bibnamefont {{Takayanagi}}},\
  }\href {\doibase 10.1007/JHEP10(2012)193} {\bibfield  {journal} {\bibinfo
  {journal} {Journal of High Energy Physics}\ }\textbf {\bibinfo {volume}
  {10}},\ \bibinfo {eid} {193} (\bibinfo {year} {2012})},\ \Eprint
  {http://arxiv.org/abs/1208.3469} {arXiv:1208.3469 [hep-th]} \BibitemShut
  {NoStop}%
\bibitem [{\citenamefont {{Pastawski}}\ \emph {et~al.}(2015)\citenamefont
  {{Pastawski}}, \citenamefont {{Yoshida}}, \citenamefont {{Harlow}},\ and\
  \citenamefont {{Preskill}}}]{Pastawski2015}%
  \BibitemOpen
  \bibfield  {author} {\bibinfo {author} {\bibfnamefont {F.}~\bibnamefont
  {{Pastawski}}}, \bibinfo {author} {\bibfnamefont {B.}~\bibnamefont
  {{Yoshida}}}, \bibinfo {author} {\bibfnamefont {D.}~\bibnamefont {{Harlow}}},
  \ and\ \bibinfo {author} {\bibfnamefont {J.}~\bibnamefont {{Preskill}}},\
  }\href {\doibase 10.1007/JHEP06(2015)149} {\bibfield  {journal} {\bibinfo
  {journal} {Journal of High Energy Physics}\ }\textbf {\bibinfo {volume}
  {6}},\ \bibinfo {eid} {149} (\bibinfo {year} {2015})},\ \Eprint
  {http://arxiv.org/abs/1503.06237} {arXiv:1503.06237 [hep-th]} \BibitemShut
  {NoStop}%
\bibitem [{\citenamefont {{Hayden}}\ \emph {et~al.}(2016)\citenamefont
  {{Hayden}}, \citenamefont {{Nezami}}, \citenamefont {{Qi}}, \citenamefont
  {{Thomas}}, \citenamefont {{Walter}},\ and\ \citenamefont
  {{Yang}}}]{Hayden2016}%
  \BibitemOpen
  \bibfield  {author} {\bibinfo {author} {\bibfnamefont {P.}~\bibnamefont
  {{Hayden}}}, \bibinfo {author} {\bibfnamefont {S.}~\bibnamefont {{Nezami}}},
  \bibinfo {author} {\bibfnamefont {X.-L.}\ \bibnamefont {{Qi}}}, \bibinfo
  {author} {\bibfnamefont {N.}~\bibnamefont {{Thomas}}}, \bibinfo {author}
  {\bibfnamefont {M.}~\bibnamefont {{Walter}}}, \ and\ \bibinfo {author}
  {\bibfnamefont {Z.}~\bibnamefont {{Yang}}},\ }\href {\doibase
  10.1007/JHEP11(2016)009} {\bibfield  {journal} {\bibinfo  {journal} {Journal
  of High Energy Physics}\ }\textbf {\bibinfo {volume} {11}},\ \bibinfo {eid}
  {9} (\bibinfo {year} {2016})},\ \Eprint {http://arxiv.org/abs/1601.01694}
  {arXiv:1601.01694 [hep-th]} \BibitemShut {NoStop}%
\bibitem [{\citenamefont {{Susskind}}(2014)}]{Susskind2014a}%
  \BibitemOpen
  \bibfield  {author} {\bibinfo {author} {\bibfnamefont {L.}~\bibnamefont
  {{Susskind}}},\ }\href@noop {} {\bibfield  {journal} {\bibinfo  {journal}
  {ArXiv e-prints}\ } (\bibinfo {year} {2014})},\ \Eprint
  {http://arxiv.org/abs/1402.5674} {arXiv:1402.5674 [hep-th]} \BibitemShut
  {NoStop}%
\bibitem [{\citenamefont {{Lehner}}\ \emph {et~al.}(2016)\citenamefont
  {{Lehner}}, \citenamefont {{Myers}}, \citenamefont {{Poisson}},\ and\
  \citenamefont {{Sorkin}}}]{Lehner2016}%
  \BibitemOpen
  \bibfield  {author} {\bibinfo {author} {\bibfnamefont {L.}~\bibnamefont
  {{Lehner}}}, \bibinfo {author} {\bibfnamefont {R.~C.}\ \bibnamefont
  {{Myers}}}, \bibinfo {author} {\bibfnamefont {E.}~\bibnamefont {{Poisson}}},
  \ and\ \bibinfo {author} {\bibfnamefont {R.~D.}\ \bibnamefont {{Sorkin}}},\
  }\href {\doibase 10.1103/PhysRevD.94.084046} {\bibfield  {journal} {\bibinfo
  {journal} {\prd}\ }\textbf {\bibinfo {volume} {94}},\ \bibinfo {eid} {084046}
  (\bibinfo {year} {2016})},\ \Eprint {http://arxiv.org/abs/1609.00207}
  {arXiv:1609.00207 [hep-th]} \BibitemShut {NoStop}%
\bibitem [{\citenamefont {{Jefferson}}\ and\ \citenamefont
  {{Myers}}(2017)}]{Jefferson2017}%
  \BibitemOpen
  \bibfield  {author} {\bibinfo {author} {\bibfnamefont {R.~A.}\ \bibnamefont
  {{Jefferson}}}\ and\ \bibinfo {author} {\bibfnamefont {R.~C.}\ \bibnamefont
  {{Myers}}},\ }\href@noop {} {\bibfield  {journal} {\bibinfo  {journal} {ArXiv
  e-prints}\ } (\bibinfo {year} {2017})},\ \Eprint
  {http://arxiv.org/abs/1707.08570} {arXiv:1707.08570 [hep-th]} \BibitemShut
  {NoStop}%
\bibitem [{\citenamefont {{Chapman}}\ \emph {et~al.}(2017)\citenamefont
  {{Chapman}}, \citenamefont {{Heller}}, \citenamefont {{Marrochio}},\ and\
  \citenamefont {{Pastawski}}}]{Chapman2017}%
  \BibitemOpen
  \bibfield  {author} {\bibinfo {author} {\bibfnamefont {S.}~\bibnamefont
  {{Chapman}}}, \bibinfo {author} {\bibfnamefont {M.~P.}\ \bibnamefont
  {{Heller}}}, \bibinfo {author} {\bibfnamefont {H.}~\bibnamefont
  {{Marrochio}}}, \ and\ \bibinfo {author} {\bibfnamefont {F.}~\bibnamefont
  {{Pastawski}}},\ }\href@noop {} {\bibfield  {journal} {\bibinfo  {journal}
  {ArXiv e-prints}\ } (\bibinfo {year} {2017})},\ \Eprint
  {http://arxiv.org/abs/1707.08582} {arXiv:1707.08582 [hep-th]} \BibitemShut
  {NoStop}%
\bibitem [{\citenamefont {{Strominger}}\ and\ \citenamefont
  {{Vafa}}(1996)}]{Strominger1996}%
  \BibitemOpen
  \bibfield  {author} {\bibinfo {author} {\bibfnamefont {A.}~\bibnamefont
  {{Strominger}}}\ and\ \bibinfo {author} {\bibfnamefont {C.}~\bibnamefont
  {{Vafa}}},\ }\href {\doibase 10.1016/0370-2693(96)00345-0} {\bibfield
  {journal} {\bibinfo  {journal} {Physics Letters B}\ }\textbf {\bibinfo
  {volume} {379}},\ \bibinfo {pages} {99} (\bibinfo {year} {1996})},\ \Eprint
  {http://arxiv.org/abs/hep-th/9601029} {hep-th/9601029} \BibitemShut {NoStop}%
\bibitem [{\citenamefont {{Callan}}\ and\ \citenamefont
  {{Maldacena}}(1996)}]{Callan1996}%
  \BibitemOpen
  \bibfield  {author} {\bibinfo {author} {\bibfnamefont {C.~G.}\ \bibnamefont
  {{Callan}}}\ and\ \bibinfo {author} {\bibfnamefont {J.~M.}\ \bibnamefont
  {{Maldacena}}},\ }\href {\doibase 10.1016/0550-3213(96)00225-8} {\bibfield
  {journal} {\bibinfo  {journal} {Nuclear Physics B}\ }\textbf {\bibinfo
  {volume} {472}},\ \bibinfo {pages} {591} (\bibinfo {year} {1996})},\ \Eprint
  {http://arxiv.org/abs/hep-th/9602043} {hep-th/9602043} \BibitemShut {NoStop}%
\bibitem [{\citenamefont {{Avery}}\ \emph {et~al.}(2010)\citenamefont
  {{Avery}}, \citenamefont {{Chowdhury}},\ and\ \citenamefont
  {{Mathur}}}]{Avery2010}%
  \BibitemOpen
  \bibfield  {author} {\bibinfo {author} {\bibfnamefont {S.~G.}\ \bibnamefont
  {{Avery}}}, \bibinfo {author} {\bibfnamefont {B.~D.}\ \bibnamefont
  {{Chowdhury}}}, \ and\ \bibinfo {author} {\bibfnamefont {S.~D.}\ \bibnamefont
  {{Mathur}}},\ }\href {\doibase 10.1007/JHEP06(2010)031} {\bibfield  {journal}
  {\bibinfo  {journal} {Journal of High Energy Physics}\ }\textbf {\bibinfo
  {volume} {6}},\ \bibinfo {eid} {31} (\bibinfo {year} {2010})},\ \Eprint
  {http://arxiv.org/abs/1002.3132} {arXiv:1002.3132 [hep-th]} \BibitemShut
  {NoStop}%
\bibitem [{\citenamefont {Levin}\ and\ \citenamefont {Wen}(2005)}]{Levin2005}%
  \BibitemOpen
  \bibfield  {author} {\bibinfo {author} {\bibfnamefont {M.~A.}\ \bibnamefont
  {Levin}}\ and\ \bibinfo {author} {\bibfnamefont {X.-G.}\ \bibnamefont
  {Wen}},\ }\href {\doibase 10.1103/PhysRevB.71.045110} {\bibfield  {journal}
  {\bibinfo  {journal} {Phys. Rev. B}\ }\textbf {\bibinfo {volume} {71}},\
  \bibinfo {pages} {045110} (\bibinfo {year} {2005})}\BibitemShut {NoStop}%
\bibitem [{\citenamefont {{K{\"o}nig}}\ \emph {et~al.}(2009)\citenamefont
  {{K{\"o}nig}}, \citenamefont {{Reichardt}},\ and\ \citenamefont
  {{Vidal}}}]{Konig2009}%
  \BibitemOpen
  \bibfield  {author} {\bibinfo {author} {\bibfnamefont {R.}~\bibnamefont
  {{K{\"o}nig}}}, \bibinfo {author} {\bibfnamefont {B.~W.}\ \bibnamefont
  {{Reichardt}}}, \ and\ \bibinfo {author} {\bibfnamefont {G.}~\bibnamefont
  {{Vidal}}},\ }\href {\doibase 10.1103/PhysRevB.79.195123} {\bibfield
  {journal} {\bibinfo  {journal} {\prb}\ }\textbf {\bibinfo {volume} {79}},\
  \bibinfo {eid} {195123} (\bibinfo {year} {2009})},\ \Eprint
  {http://arxiv.org/abs/0806.4583} {arXiv:0806.4583 [cond-mat.str-el]}
  \BibitemShut {NoStop}%
\bibitem [{\citenamefont {{Gu}}\ \emph {et~al.}(2009)\citenamefont {{Gu}},
  \citenamefont {{Levin}}, \citenamefont {{Swingle}},\ and\ \citenamefont
  {{Wen}}}]{Gu2009}%
  \BibitemOpen
  \bibfield  {author} {\bibinfo {author} {\bibfnamefont {Z.-C.}\ \bibnamefont
  {{Gu}}}, \bibinfo {author} {\bibfnamefont {M.}~\bibnamefont {{Levin}}},
  \bibinfo {author} {\bibfnamefont {B.}~\bibnamefont {{Swingle}}}, \ and\
  \bibinfo {author} {\bibfnamefont {X.-G.}\ \bibnamefont {{Wen}}},\ }\href
  {\doibase 10.1103/PhysRevB.79.085118} {\bibfield  {journal} {\bibinfo
  {journal} {\prb}\ }\textbf {\bibinfo {volume} {79}},\ \bibinfo {eid} {085118}
  (\bibinfo {year} {2009})},\ \Eprint {http://arxiv.org/abs/0809.2821}
  {arXiv:0809.2821 [cond-mat.str-el]} \BibitemShut {NoStop}%
\bibitem [{\citenamefont {Vidal}(2007)}]{Vidal2007}%
  \BibitemOpen
  \bibfield  {author} {\bibinfo {author} {\bibfnamefont {G.}~\bibnamefont
  {Vidal}},\ }\href {\doibase 10.1103/physrevlett.98.070201} {\bibfield
  {journal} {\bibinfo  {journal} {Phys. Rev. Lett.}\ }\textbf {\bibinfo
  {volume} {98}} (\bibinfo {year} {2007}),\
  10.1103/physrevlett.98.070201}\BibitemShut {NoStop}%
\bibitem [{\citenamefont {Jordan}\ \emph {et~al.}(2008)\citenamefont {Jordan},
  \citenamefont {Or\'us}, \citenamefont {Vidal}, \citenamefont {Verstraete},\
  and\ \citenamefont {Cirac}}]{Jordan2008}%
  \BibitemOpen
  \bibfield  {author} {\bibinfo {author} {\bibfnamefont {J.}~\bibnamefont
  {Jordan}}, \bibinfo {author} {\bibfnamefont {R.}~\bibnamefont {Or\'us}},
  \bibinfo {author} {\bibfnamefont {G.}~\bibnamefont {Vidal}}, \bibinfo
  {author} {\bibfnamefont {F.}~\bibnamefont {Verstraete}}, \ and\ \bibinfo
  {author} {\bibfnamefont {J.~I.}\ \bibnamefont {Cirac}},\ }\href {\doibase
  10.1103/physrevlett.101.250602} {\bibfield  {journal} {\bibinfo  {journal}
  {Phys. Rev. Lett.}\ }\textbf {\bibinfo {volume} {101}} (\bibinfo {year}
  {2008}),\ 10.1103/physrevlett.101.250602}\BibitemShut {NoStop}%
\bibitem [{\citenamefont {{Yan}}\ \emph {et~al.}(2011)\citenamefont {{Yan}},
  \citenamefont {{Huse}},\ and\ \citenamefont {{White}}}]{Yan2011}%
  \BibitemOpen
  \bibfield  {author} {\bibinfo {author} {\bibfnamefont {S.}~\bibnamefont
  {{Yan}}}, \bibinfo {author} {\bibfnamefont {D.~A.}\ \bibnamefont {{Huse}}}, \
  and\ \bibinfo {author} {\bibfnamefont {S.~R.}\ \bibnamefont {{White}}},\
  }\href {\doibase 10.1126/science.1201080} {\bibfield  {journal} {\bibinfo
  {journal} {Science}\ }\textbf {\bibinfo {volume} {332}},\ \bibinfo {pages}
  {1173} (\bibinfo {year} {2011})},\ \Eprint {http://arxiv.org/abs/1011.6114}
  {arXiv:1011.6114 [cond-mat.str-el]} \BibitemShut {NoStop}%
\bibitem [{\citenamefont {{Jiang}}\ \emph {et~al.}(2012)\citenamefont
  {{Jiang}}, \citenamefont {{Wang}},\ and\ \citenamefont
  {{Balents}}}]{Jiang2012}%
  \BibitemOpen
  \bibfield  {author} {\bibinfo {author} {\bibfnamefont {H.-C.}\ \bibnamefont
  {{Jiang}}}, \bibinfo {author} {\bibfnamefont {Z.}~\bibnamefont {{Wang}}}, \
  and\ \bibinfo {author} {\bibfnamefont {L.}~\bibnamefont {{Balents}}},\ }\href
  {\doibase 10.1038/nphys2465} {\bibfield  {journal} {\bibinfo  {journal}
  {Nature Physics}\ }\textbf {\bibinfo {volume} {8}},\ \bibinfo {pages} {902}
  (\bibinfo {year} {2012})},\ \Eprint {http://arxiv.org/abs/1205.4289}
  {arXiv:1205.4289 [cond-mat.str-el]} \BibitemShut {NoStop}%
\bibitem [{\citenamefont {{He}}\ \emph {et~al.}(2017)\citenamefont {{He}},
  \citenamefont {{Zaletel}}, \citenamefont {{Oshikawa}},\ and\ \citenamefont
  {{Pollmann}}}]{He2017}%
  \BibitemOpen
  \bibfield  {author} {\bibinfo {author} {\bibfnamefont {Y.-C.}\ \bibnamefont
  {{He}}}, \bibinfo {author} {\bibfnamefont {M.~P.}\ \bibnamefont {{Zaletel}}},
  \bibinfo {author} {\bibfnamefont {M.}~\bibnamefont {{Oshikawa}}}, \ and\
  \bibinfo {author} {\bibfnamefont {F.}~\bibnamefont {{Pollmann}}},\ }\href
  {\doibase 10.1103/PhysRevX.7.031020} {\bibfield  {journal} {\bibinfo
  {journal} {Physical Review X}\ }\textbf {\bibinfo {volume} {7}},\ \bibinfo
  {eid} {031020} (\bibinfo {year} {2017})},\ \Eprint
  {http://arxiv.org/abs/1611.06238} {arXiv:1611.06238 [cond-mat.str-el]}
  \BibitemShut {NoStop}%
\bibitem [{\citenamefont {Verstraete}\ and\ \citenamefont
  {Cirac}(2005)}]{Verstraete2005}%
  \BibitemOpen
  \bibfield  {author} {\bibinfo {author} {\bibfnamefont {F.}~\bibnamefont
  {Verstraete}}\ and\ \bibinfo {author} {\bibfnamefont {J.~I.}\ \bibnamefont
  {Cirac}},\ }\href {\doibase 10.1088/1742-5468/2005/09/P09012} {\bibfield
  {journal} {\bibinfo  {journal} {J. Stat. Mech.}\ }\textbf {\bibinfo {volume}
  {0509}},\ \bibinfo {pages} {P09012} (\bibinfo {year} {2005})},\ \Eprint
  {http://arxiv.org/abs/cond-mat/0508353v3} {cond-mat/0508353v3} \BibitemShut
  {NoStop}%
\bibitem [{\citenamefont {{Bohnet}}\ \emph {et~al.}(2016)\citenamefont
  {{Bohnet}}, \citenamefont {{Sawyer}}, \citenamefont {{Britton}},
  \citenamefont {{Wall}}, \citenamefont {{Rey}}, \citenamefont {{Foss-Feig}},\
  and\ \citenamefont {{Bollinger}}}]{Bohnet2015}%
  \BibitemOpen
  \bibfield  {author} {\bibinfo {author} {\bibfnamefont {J.~G.}\ \bibnamefont
  {{Bohnet}}}, \bibinfo {author} {\bibfnamefont {B.~C.}\ \bibnamefont
  {{Sawyer}}}, \bibinfo {author} {\bibfnamefont {J.~W.}\ \bibnamefont
  {{Britton}}}, \bibinfo {author} {\bibfnamefont {M.~L.}\ \bibnamefont
  {{Wall}}}, \bibinfo {author} {\bibfnamefont {A.~M.}\ \bibnamefont {{Rey}}},
  \bibinfo {author} {\bibfnamefont {M.}~\bibnamefont {{Foss-Feig}}}, \ and\
  \bibinfo {author} {\bibfnamefont {J.~J.}\ \bibnamefont {{Bollinger}}},\
  }\href {\doibase 10.1126/science.aad9958} {\bibfield  {journal} {\bibinfo
  {journal} {Science}\ }\textbf {\bibinfo {volume} {352}},\ \bibinfo {pages}
  {1297} (\bibinfo {year} {2016})},\ \Eprint {http://arxiv.org/abs/1512.03756}
  {arXiv:1512.03756 [quant-ph]} \BibitemShut {NoStop}%
\bibitem [{\citenamefont {{Zhang}}\ \emph {et~al.}(2017)\citenamefont
  {{Zhang}}, \citenamefont {{Pagano}}, \citenamefont {{Hess}}, \citenamefont
  {{Kyprianidis}}, \citenamefont {{Becker}}, \citenamefont {{Kaplan}},
  \citenamefont {{Gorshkov}}, \citenamefont {{Gong}},\ and\ \citenamefont
  {{Monroe}}}]{Zhang2017}%
  \BibitemOpen
  \bibfield  {author} {\bibinfo {author} {\bibfnamefont {J.}~\bibnamefont
  {{Zhang}}}, \bibinfo {author} {\bibfnamefont {G.}~\bibnamefont {{Pagano}}},
  \bibinfo {author} {\bibfnamefont {P.~W.}\ \bibnamefont {{Hess}}}, \bibinfo
  {author} {\bibfnamefont {A.}~\bibnamefont {{Kyprianidis}}}, \bibinfo {author}
  {\bibfnamefont {P.}~\bibnamefont {{Becker}}}, \bibinfo {author}
  {\bibfnamefont {H.}~\bibnamefont {{Kaplan}}}, \bibinfo {author}
  {\bibfnamefont {A.~V.}\ \bibnamefont {{Gorshkov}}}, \bibinfo {author}
  {\bibfnamefont {Z.-X.}\ \bibnamefont {{Gong}}}, \ and\ \bibinfo {author}
  {\bibfnamefont {C.}~\bibnamefont {{Monroe}}},\ }\href@noop {} {\bibfield
  {journal} {\bibinfo  {journal} {ArXiv e-prints}\ } (\bibinfo {year}
  {2017})},\ \Eprint {http://arxiv.org/abs/1708.01044} {arXiv:1708.01044
  [quant-ph]} \BibitemShut {NoStop}%
\bibitem [{\citenamefont {{Bernien}}\ \emph {et~al.}(2017)\citenamefont
  {{Bernien}}, \citenamefont {{Schwartz}}, \citenamefont {{Keesling}},
  \citenamefont {{Levine}}, \citenamefont {{Omran}}, \citenamefont {{Pichler}},
  \citenamefont {{Choi}}, \citenamefont {{Zibrov}}, \citenamefont {{Endres}},
  \citenamefont {{Greiner}}, \citenamefont {{Vuleti{\'c}}},\ and\ \citenamefont
  {{Lukin}}}]{Bernien2017}%
  \BibitemOpen
  \bibfield  {author} {\bibinfo {author} {\bibfnamefont {H.}~\bibnamefont
  {{Bernien}}}, \bibinfo {author} {\bibfnamefont {S.}~\bibnamefont
  {{Schwartz}}}, \bibinfo {author} {\bibfnamefont {A.}~\bibnamefont
  {{Keesling}}}, \bibinfo {author} {\bibfnamefont {H.}~\bibnamefont
  {{Levine}}}, \bibinfo {author} {\bibfnamefont {A.}~\bibnamefont {{Omran}}},
  \bibinfo {author} {\bibfnamefont {H.}~\bibnamefont {{Pichler}}}, \bibinfo
  {author} {\bibfnamefont {S.}~\bibnamefont {{Choi}}}, \bibinfo {author}
  {\bibfnamefont {A.~S.}\ \bibnamefont {{Zibrov}}}, \bibinfo {author}
  {\bibfnamefont {M.}~\bibnamefont {{Endres}}}, \bibinfo {author}
  {\bibfnamefont {M.}~\bibnamefont {{Greiner}}}, \bibinfo {author}
  {\bibfnamefont {V.}~\bibnamefont {{Vuleti{\'c}}}}, \ and\ \bibinfo {author}
  {\bibfnamefont {M.~D.}\ \bibnamefont {{Lukin}}},\ }\href@noop {} {\bibfield
  {journal} {\bibinfo  {journal} {ArXiv e-prints}\ } (\bibinfo {year}
  {2017})},\ \Eprint {http://arxiv.org/abs/1707.04344} {arXiv:1707.04344
  [quant-ph]} \BibitemShut {NoStop}%
\end{thebibliography}%

\appendix

\end{document}